\pdfoutput=1  
\documentclass[11pt,a4paper]{article}

\usepackage[T1]{fontenc}
\usepackage[utf8]{inputenc}
\usepackage[english]{babel}
\usepackage[a4paper,margin=25mm]{geometry}
\usepackage{amsmath,amssymb,amsthm}
\usepackage{graphicx}
\usepackage{tikz}
\usetikzlibrary{arrows.meta,positioning}
\usepackage{booktabs}
\usepackage{multirow}
\usepackage{longtable}
\usepackage{array}
\usepackage{microtype}
\usepackage{xcolor}
\usepackage{xurl}
\usepackage[hidelinks,colorlinks=false]{hyperref}
\hypersetup{
  pdftitle={SHARD: cell-keyed residual splitting for alignment-resistant private dense retrieval},
  pdfauthor={Sergey Kurilenko}
}
\usepackage{enumitem}
\usepackage{caption}
\usepackage{float}

\usepackage{placeins}
\graphicspath{{figs/}}
\newtheorem{proposition}{Proposition}
\newtheorem{hypothesis}{Empirical Hypothesis}

\title{\textbf{SHARD: cell-keyed residual splitting for
alignment-resistant private dense retrieval}}
\author{Sergey Kurilenko\\
\small Moscow Institute of Physics and Technology\\
\small \texttt{sergkurilenko@gmail.com}}
\date{\today}

\begin{document}
\maketitle

\begin{abstract}
\noindent
Dense retrieval systems expose document geometry when vector stores are
compromised, and a global protective transform can often be aligned from
known pairs. We study \textsc{Shard}, which splits PCA coordinates into a
short routing prefix and a residual protected by independent cell-local
orthogonal keys. It supports CKKS ciphertext--plaintext reranking but is
evaluated as a leakage trade-off, not a cryptographic document-privacy
guarantee.

Corrected scoring uses $V^\top(x-\mu)$ for documents and $V^\top q$ for
queries, preserving raw ranking up to a query-dependent constant. Across ten
BEIR/MIRACL configurations it reproduces raw nDCG@10 and recall, whereas
centring both sides loses up to $0.080$ nDCG. Cell keys spread diffuse
known-pair evidence across compartments, but minimum-norm alignment recovers
useful signal far below full key rank, so there is no hard de-anonymisation
threshold.

Real CKKS confirms a maximum score error of $2.29\times10^{-6}$ and no top-1
flips; block packing cuts query upload by $74$--$87\%$ but raises in-process
p50 latency by $14$--$26\%$. In a strengthened GTR case, an unknown key lowers
token-F1 from $.665$ to $.242$; a wide prefix and eight pairs restore much.
Under $25$--$90\%$ release overlap, the unchanged prefix and clean residual
norm link persistent rows with R@1 at least $0.9996$, although cell-Gram
linkage degrades under heavy churn. A formally calibrated Gaussian release
gives nDCG@10 at most $0.011$ at $\varepsilon=1$; its only three strict
utility matches occur at $\varepsilon=32768$ with linkage R@1 at least
$0.995$. \textsc{Shard} therefore preserves retrieval and compartmentalises
alignment evidence, but does not provide DP, unlinkability, or cancellable
templates.
\end{abstract}


\section{Introduction}

The privacy of vector representations of text---``embeddings''---has become
a first-class concern in production retrieval systems. Dense retrievers
(DPR \cite{ref_dpr}, ColBERT \cite{ref_colbert}) and state-of-the-art
encoders (Sentence-BERT \cite{ref_sbert}, multilingual-E5 \cite{ref_e5},
GTR \cite{ref_gtr}, BGE-M3 \cite{ref_bgem3}) push the per-document
representation through a
$\sim 10^2$-dimensional bottleneck that nevertheless preserves enough
information to reconstruct the original text with high BLEU~\cite{ref_bleu}.
\textsc{Vec2Text} \cite{ref_vec2text} reports
$\mathrm{BLEU} \approx 97.3\%$ on $32$-token inputs when the attacker has
query access to the encoder; \textsc{GEIA}~\cite{ref_geia} and
\textsc{TEIA}~\cite{ref_teia} show that even snapshot attacks
(an adversary with read access to the vector database) recover personally
identifiable information at high rates. For Retrieval-Augmented Generation
pipelines whose vector stores typically contain customer support tickets,
internal corporate documents and personal data, this rule of thumb---``a
leak of the embedding is a leak of the text''---is now well documented
\cite{ref_rag,ref_zeng_rag}.

The defensive landscape is bracketed by two extremes that, individually,
do not satisfy industrial requirements. Differential privacy
\cite{ref_dwork,ref_dpsgd} adds calibrated noise per-record but, in the
naive coordinate-Gaussian formulation on retrieval-quality embeddings, it
faces a difficult operating-point choice: mild noise preserves much of the
nearest-neighbour geometry, while stronger noise destroys ranking quality
long before the privacy budget reaches a level considered useful.
Fully homomorphic encryption \cite{ref_ckks,ref_openfhe}
allows arbitrary computation on ciphertexts, but the
ciphertext-ciphertext (ct-ct) regime is too expensive for top-$k$ search
across $10^6$ documents.

\paragraph{The weak axis of global-linear defences.} Between the two
extremes sits a popular lightweight middle ground: compress the collection
onto a data-dependent SVD subspace, apply a single secret orthogonal
rotation, publish an approximate-nearest-neighbour (ANN) index, and rerank
the encrypted query under CKKS. We first measure this global-linear stack
and confirm that its protection lives on \emph{one weak axis}---the
protected store is a single globally-aligned geometry. Concretely
(Section~\ref{sec:experiments}): a known-plaintext orthogonal Procrustes
attack recovers the secret rotation with about $k=d/2$ anchors, after which
an overlapping reference corpus yields near-exact paragraph recovery
($99.8\%$ top-1); the public product-quantisation index preserves a mean
cosine of $0.95$ and $67\%$ of exact top-10 neighbours; and a truncated
rerank can lose quality on real IR. A corrected scoring audit in this paper
separates two effects that the earlier implementation conflated: centring the
query as well as the documents costs $0.9$--$7.4$ nDCG points, while the
remaining half-PCA truncation gap ranges from approximately $0$ to $5.3$
points across the evaluated BEIR and MIRACL cells. These are
exactly the weaknesses that modern few-shot alignment \cite{ref_algen},
zero-shot inversion \cite{ref_zsinvert} and unsupervised cross-space
translation \cite{ref_vec2vec} attacks exploit: a single global map is, in
the end, alignable.

\paragraph{\textsc{Shard}: an attack-aware geometric transform.} We
introduce \textsc{Shard}, a family of retrieval-preserving protective
transforms designed against the alignment and index-leakage attack surface
rather than against a distortion surrogate such as $\sigma_{rec}$ (which,
as we and others show, is not a privacy metric). The construction
(Section~\ref{sec:shard}) keeps the CKKS/ct-pt machinery of the baseline
but replaces its single global geometry. The construction itself is quick to state. We rotate the
centred embedding and split it into a short \emph{public} prefix $u$---the
top-$d_{\mathrm{pub}}$ PCA directions, which carry coarse stage-1
retrieval---and a \emph{private} residual $r$. The residual is sharded into
$C$ coarse cells defined by $u$, and each cell is turned by its own secret
orthogonal key $H_c$. The experiments sample dense Haar-like keys by Gaussian
QR, so that the server only stores $z_i = H_{c(i)} r_i$. Reranking then happens under CKKS in
the residual space, where the per-cell keys cancel,
$\langle H_c r_q, z_i\rangle=\langle r_q, r_i\rangle$ in plaintext algebra;
CKKS then evaluates the same identity approximately. The single knob $C$
interpolates between a full-dimensional single-residual-key control
($C{=}1$) and per-document micro-keys
($C{=}N$).

The geometric evaluation begins with two constructive findings. First, the
corrected full-dimensional split score preserves raw
ranking on all ten BEIR/MIRACL encoder--dataset cells; in the plaintext
experiment, the remaining measured utility gap comes from stage-1 shortlist
recall. Second, cell keys compartmentalise
known-pair evidence. A diffuse leak gives each cell fewer anchors, so the
recoverable anchor-span rank and the resulting cosine/re-identification curve
grow more slowly as $C$ increases. This effect is continuous, however:
minimum-norm linear alignment extracts useful signal far below full rank, so
there is no $d_{\mathrm{priv}}$ privacy barrier and no claim that subthreshold
records remain unidentified.

The limiting findings are just as important. The short prefix exposes less global neighbour structure than the
baseline's public half-space, but the server still sees a meaningful coarse
index and exact within-cell Gram geometry. Independent re-keying therefore
does not make snapshots unlinkable. At $N=10{,}000$, a row-wise cross-key
cosine test is near chance, yet norm assignment, Gram signatures and the
unchanged prefix link almost every persistent record. Insert/delete churn
weakens Gram matching at low overlap, but prefix and clean residual-norm R@1
remain above $0.9996$; micro-keys remove shared residual edges without closing
those two channels.

The systems and privacy baselines add two less convenient but important
findings. Real CKKS/block-SIMD execution preserves ranking to numerical
precision and cuts query upload by $74$--$87\%$, yet increases median latency
in the tested TenSEAL layout because every candidate still returns a separate
ciphertext. A formally calibrated Gaussian release, meanwhile, has no
evaluated point that combines SHARD-matched utility with resistance to clean
native-gallery linkage. In a strengthened GTR/Vec2Text outcome test, hiding
the cell key lowers mean token-F1 from $.665$ to $.242$, but the wide public
prefix and eight known pairs restore much of that loss. These results do not
turn SHARD into DP or into a
cryptographic document store; they make the boundary of the contribution
measurable.

\noindent
The global-linear baseline characterisation (alignment, PQ leakage,
BEIR-truncation loss, and the calibrated-noise diagnostic that
$\sigma_{rec}$ is not a privacy metric) is retained in
Section~\ref{sec:experiments} as the foil that motivates \textsc{Shard},
together with the CKKS ct-pt path that both schemes share. The earlier short
paper introduced a random-projection hybrid~\cite{ref_vkit}; the subsequent
global-linear study replaced that projection with the disclosed SVD/PQ
geometry and supplied the reproducible CKKS parameter selection and leakage
audit~\cite{ref_alpha}. \textsc{Shard} is the attack-aware successor that
replaces the single global secret axis.

\paragraph{Contributions relative to prior work.} The global-linear
baseline---SVD truncation, a secret rotation, a public ANN index and ct-pt
CKKS reranking---together with the reproducible CKKS parameter selection are
\emph{not} new here: the exact comparison system is the prior global-linear
study~\cite{ref_alpha}, preceded by the random-projection prototype in
\cite{ref_vkit}. We reproduce it only as the comparison setting. Relative to
that prior work, this
paper begins with the \textsc{Shard} cell-keyed residual split and its
rank-preserving scoring correction (Section~\ref{sec:shard}). It then follows
the construction through the failure modes that determine its honest scope:
anchor-span recovery below full key rank (Section~\ref{ssec:shard-align});
public-prefix, norm and within-cell leakage, including the micro-key limit and
release churn (Sections~\ref{ssec:shard-leak} and~\ref{ssec:shard-micro}); and
an overlap-reference lookup (Section~\ref{ssec:shard-ref}). The empirical case
is completed by a two-encoder, ten-cell BEIR/MIRACL utility audit, an actual
TenSEAL block-SIMD implementation, a formally accounted Gaussian-release
baseline, and a learned text-level outcome. Everything else is reproduced
baseline context, condensed and attributed accordingly.

\paragraph{Hypotheses.} The study is organised around four testable
hypotheses. \textbf{(H1)} Under diffuse disclosure, increasing $C$ lowers
the in-cell anchor rank and shifts partial-alignment recovery curves, without
creating a hard full-rank threshold. \textbf{(H2)} The corrected
full-dimensional split score preserves raw ranking, while half-PCA truncation
may still lose graded retrieval quality. \textbf{(H3)} In the measured
settings, the short prefix exposes less exact-neighbour structure than a
public half-space PQ index, while shortlist recall is evaluated separately.
\textbf{(H4)} Pairwise cross-key cosine is an
insufficient unlinkability test because prefix, norm and Gram invariants can
link re-keyed snapshots. The experiments of
Section~\ref{sec:experiments} test each in turn. These hypotheses concern the
geometric construction; the CKKS, formal-DP and churn studies are
complementary implementation and robustness audits.

\paragraph{Security-claim discipline.} \textsc{Shard} is an attack-aware
\emph{geometric} defence, not a cryptographic privacy guarantee, and every
claim below should be read as holding \emph{in the measured setting and
against the studied attacker families}. The earlier baseline prototype uses
CKKS to protect numerical query values and returned scores. Here SHARD also
reuses that path in an actual TenSEAL implementation, so numerical error,
serialization volume and phase latency are measured rather than inferred.
This is still a single-process workstation benchmark, not a networked-service
claim. The document side is geometrically transformed, not proven secure, and
access patterns together with an overlapping public reference corpus remain
exposed channels that we measure rather than close.

\paragraph{Notation.} To avoid the overloading of the symbol $N$, we use
$N_{\mathrm{poly}}$ for the CKKS polynomial-modulus degree,
$N_{\mathrm{docs}}$ for the corpus size, and $K_{\mathrm{cands}}$ for the
stage-1 short-list length throughout.

\paragraph{An elementary proxy bound.} We state a one-line projection
bound (Proposition~\ref{lemma:proj-l2}): any decoder whose image is
constrained to $\mathrm{span}(V_k)$ has $L_2$ reconstruction error at
least $\|\mathbf{x}_\perp\|_2$. This is a Pythagorean fact, not a security
theorem; we use it only to make $\sigma_{rec}$ a precise distortion proxy
within that decoder class. Its empirical relation to the BLEU of an
off-the-shelf \textsc{Vec2Text} attack is reported as Empirical
Hypothesis~\ref{hyp:bleu}.

\paragraph{Roadmap.} Section~\ref{sec:related} reviews related
literature; Section~\ref{sec:threat} states the threat model;
Section~\ref{sec:method} describes the global-linear baseline (the foil)
and the shared CKKS/ct-pt machinery;
Section~\ref{sec:lemma} an elementary projection bound;
Section~\ref{sec:shard} introduces \textsc{Shard};
Section~\ref{sec:experiments} reports the baseline characterisation and the
expanded \textsc{Shard} evaluation: corrected utility, partial alignment,
public-index leakage, measured CKKS, cross-release linkage with churn, a
formal Gaussian release, learned text inversion, and reference lookup;
Section~\ref{sec:discussion} discusses limitations and open work.
Appendix~\ref{app:baseline} condenses the baseline inversion analyses
that overlap with the global-linear paper~\cite{ref_alpha}.

\section{Related Work}\label{sec:related}

\paragraph{Embedding inversion.} \textsc{Vec2Text} \cite{ref_vec2text}
trains an encoder--decoder model that iteratively rewrites a candidate
hypothesis until its embedding matches a target vector; on
$32$-token \textsc{MS-MARCO} inputs the attack achieves
$\mathrm{BLEU} \approx 97.3\%$ and $\approx 92\%$ exact recovery.
\textsc{GEIA} \cite{ref_geia} instead trains a generative decoder that maps
a target embedding to a prompt and reconstructs the whole sentence in a
single pass. \textsc{TEIA} \cite{ref_teia} removes the need to query the
victim encoder, training a surrogate inverter that transfers across
embedding models. Earlier work of Song and
Raghunathan~\cite{ref_song_leak} already established that embedding models
leak substantial information about their inputs. Adjacent work on
membership-inference \cite{ref_shokri} and on extracting training
texts from large language models \cite{ref_carlini} confirms that the
risk surface of dense retrievers is broader than the embedding-inversion
literature alone suggests. The 2025 generation of attacks further
weakens any defence that relies only on an unknown embedding orientation:
\textsc{ALGEN} aligns a victim embedding space to an attack space with
few known pairs and then generates text \cite{ref_algen};
\textsc{ZSInvert} removes the need for an embedding-specific inversion
model by using a universal zero-shot decoding strategy
\cite{ref_zsinvert}; and \textsc{LAGO} extends few-shot inversion to
cross-lingual settings by coupling alignments across related languages
\cite{ref_lago}. These papers motivate the alignment experiment in
Section~\ref{ssec:alignment-pq}: if a secret rotation can be absorbed by
a small alignment matrix, it must not be described as a security
guarantee.

\paragraph{Encrypted search and leakage.} A long cryptographic line studies
search over encrypted data and the leakage it admits by construction.
Searchable symmetric encryption formalised the efficiency/leakage
trade-off~\cite{ref_sse}, and structured encryption generalised it to
encrypted data structures with controlled disclosure~\cite{ref_structenc}.
A parallel line of leakage-abuse cryptanalysis shows that access-pattern and
co-occurrence leakage alone can reconstruct queries and
data~\cite{ref_leakageabuse,ref_accesspattern}. We position \textsc{Shard}
within this framing rather than against it: CKKS reranking hides the value
channel, but the public stage-1 index and the per-query access pattern are
leakage that remains \emph{by construction}, which is why we measure the
public-index neighbour leakage explicitly and flag access patterns as an
open channel for composition with a PIR/ORAM-style primitive. A recent
reproducibility study of \textsc{Vec2Text}~\cite{ref_v2t_repro} corroborates
that inversion is strong under ideal conditions yet sensitive to sequence
length and to simple quantisation/noise, consistent with our weak-attacker
calibration.

\paragraph{Lightweight defences and transformed representations.}
Recent defences increasingly treat embeddings as representations whose
geometry can be compressed, transformed or regularised rather than merely
noised. \textsc{EGuard} uses a projection network with mutual-information
regularisation to reduce inversion leakage \cite{ref_eguard};
\textsc{SPARSE} studies concept-aware privacy mechanisms and anisotropic
noise for embedding inversion resistance \cite{ref_sparse}; and
nonparametric variational DP work adds bottleneck/noise mechanisms at
the representation level \cite{ref_nvdp}. Matryoshka-style representation
learning shows that retrieval utility can be deliberately concentrated in
prefix subspaces \cite{ref_mrl}, while recent cross-space translation
work argues that embedding spaces share enough geometry for transfer
maps to be useful \cite{ref_vec2vec}. \textsc{Shard} builds on both
observations: it uses a Matryoshka-style prefix as a deliberately coarse
public channel, precisely because the cross-space-translation literature
shows that a single global geometry is alignable.

\paragraph{Positioning against the 2025 defense and attack lines.} Two recent
defenses target different channels than \textsc{Shard}.
\textsc{STEER}~\cite{ref_steer} is \emph{query-side}: the client learns a
one-off transform that maps its query into the server's embedding space, so
the server cannot invert the query---but it does not protect the stored
documents and its cross-model alignment is approximate.
\textsc{EntroGuard}~\cite{ref_entroguard} (like \textsc{EGuard}) is a
\emph{perturbation} defense: it raises the entropy of inverted text under a
distortion bound, protecting embeddings against trained inversion but at a
lossy, non-exact reranking cost. On the attack side, \textsc{ZSInvert} and
the universal geometry of \textsc{vec2vec} make a leaked store
inversion-prone without victim-specific training---precisely the alignment
surface \textsc{Shard} hardens. Table~\ref{tab:defense-compare} places these
lines on common axes. \textsc{Shard} combines a rank-correct geometric split
with a CKKS-compatible query path, but its document view is linkable through
the prefix, norms and cell geometry. We therefore avoid an ``only'' or
end-to-end privacy claim.

\begin{table}[tbp]
\centering
\caption{Positioning against recent privacy-preserving retrieval defenses.
$\checkmark$ provided, $\times$ not addressed, $\sim$ partial/approximate.
Rank fidelity refers to plaintext scoring; CKKS remains approximate.}
\label{tab:defense-compare}
\scriptsize
\setlength{\tabcolsep}{2.2pt}
\begin{tabular}{@{}lccccl@{}}
\toprule
Defense & Query priv. & Doc.\ priv. & Rank fidelity & Limited leak & Attacker / cost \\ \midrule
\textbf{\textsc{Shard}} (ours) & $\sim$ CKKS path & $\sim$ keyed & $\checkmark$ plain. & $\times$ & prefix/norm/Gram; active cells \\
\textsc{STEER}~\cite{ref_steer} & $\checkmark$ align. & $\times$ & $\sim$ approx. & --- & query inversion; one-off \\
\textsc{EntroGuard}~\cite{ref_entroguard} & $\times$ & $\sim$ perturb. & $\times$ & $\sim$ & trained EIA; low \\
\textsc{EGuard}~\cite{ref_eguard} & $\times$ & $\sim$ learned & $\times$ & $\sim$ & trained EIA; low \\
PQ-only / quantization & $\times$ & $\sim$ lossy & $\times$ & $\sim$ & reconstruction; low \\
legacy uncalibrated Gaussian & $\times$ & $\sim$ noise & $\times$ & $\times$ & distortion diagnostic; low \\
analytic Gaussian release~\cite{ref_analytic_gaussian} & $\times$ & $\checkmark$ DP & $\times$ & $\sim$ & formal content DP; high utility cost \\
\bottomrule
\end{tabular}
\end{table}

\paragraph{Template criteria and decoupled encoders.}
Cancellable biometrics \cite{ref_cancellable} require renewability and
unlinkability, criteria that orthogonal re-keying alone does not satisfy here:
Experiments~25 and~28 link re-keyed records through invariant features, both
with and without insert/delete churn. Work on
retrieval with asymmetric
document/query encoders \cite{ref_dpr,ref_colbert} shows that the
document and query sides need not share one geometry---the premise that lets
\textsc{Shard} transform the stored side and the scoring query consistently.
Our contribution is the resulting compartmentalisation and, equally, the
measurement of where it fails the stronger template criteria.

\paragraph{Differential privacy for embeddings.} Lyu et al.~\cite{ref_lyu}
investigate per-coordinate Gaussian noise on dense representations. Our
formal replacement-adjacency baseline uses analytic Gaussian calibration and
makes the operating-point tension concrete: at $\varepsilon=1$, nDCG@10 is at
most $0.011$, while every strict match to corrected SHARD utility occurs at
$\varepsilon=32768$ with native-gallery linkage R@1 at least $0.995$.
Random-projection-based privacy
\cite{ref_jl,ref_liu,ref_kenthapadi} predates modern dense retrieval; the
data-dependent SVD projection used here is closer to a noise-removal
technique than to a privacy-preserving DP mechanism, although
$\sigma_{rec}$ provides a useful proxy quantity.

\paragraph{Homomorphic encryption.} The CKKS scheme~\cite{ref_ckks}
introduced approximate-arithmetic FHE; OpenFHE \cite{ref_openfhe} and
TenSEAL~\cite{ref_tenseal} provide production-grade implementations.
The bootstrapping cost \cite{ref_ckks_bootstrap} pushes practitioners
to leveled circuits; specialised compilers \cite{ref_eva,ref_chet}
optimise scale and modulus chains.
GPU acceleration of CKKS bootstrap \cite{ref_jung_gpu} and Intel HEXL
\cite{ref_hexl} have made larger leveled circuits feasible.

\paragraph{Hybrid privacy-preserving retrieval.} Tiptoe \cite{ref_tiptoe}
combines clustering with PIR for query privacy. SealPIR
\cite{ref_sealpir}, OnionPIR \cite{ref_onionpir} and SimplePIR
\cite{ref_simplepir} provide single-server PIR with practical query sizes.
The architecture in this paper is complementary: ct-pt CKKS reranking
takes care of the value side of the leakage, while access patterns
remain a separate problem for which PIR/ORAM-like primitives are an
appropriate composition.

\section{Threat Model}\label{sec:threat}

The global-linear baseline and \textsc{Shard} share one threat model, which
we fix here. It is asymmetric---a \emph{trusted client} (the data owner), an
\emph{honest-but-curious server}, and a \emph{non-adaptive} baseline
attacker, strengthened by the \emph{known-plaintext} attacker we use
throughout the alignment analysis. This is a fat-client, data-owner setting:
the client runs the encoder, holds all keys, performs stage-1 retrieval
locally, and the server only reranks under CKKS---a scoping assumption, not
the only possible deployment. Tables~\ref{tab:threat} and~\ref{tab:claims}
map the \emph{baseline}'s coverage and claims (the foil);
Table~\ref{tab:shard-claims} does the same for \textsc{Shard} (the
contribution, Section~\ref{sec:shard}).

The setting has three parties. The \emph{client} is the data owner: it runs
the encoder $E$, generates the secret keys $\mu, V_k, R, sk_{\mathrm{CKKS}}$
and never shares them, so it is trusted both in software and in key
management. The \emph{server} runs the search service over the rotated
database $E_{\mathrm{rot}} \in \mathbb{R}^{N_{\mathrm{docs}} \times k}$ and the
public PQ artefact; it is honest-but-curious, following the protocol
faithfully while trying to read textual information out of everything it
stores and every query it processes. The \emph{attacker}, finally, is at
minimum non-adaptive---an off-the-shelf \textsc{Vec2Text} model with no
knowledge of $R$---but against both schemes we grant it considerably more: a
known-plaintext alignment observer (diffuse and targeted), public-index
leakage, reference-corpus lookup, rank-deficient Procrustes, minimum-norm and
regularised linear maps, and cross-release matching from prefix, norm and
Gram invariants (Sections~\ref{ssec:shard-align} and
\ref{ssec:shard-micro}). Malicious-server behaviour remains out of scope;
the learned text-inversion study is reported separately and does not turn the
geometric tests into a universal adaptive-attacker guarantee.

\paragraph{Three distinct privacy notions.} A recurring source of
confusion in hybrid designs is the conflation of \emph{what} is
protected. We therefore separate three independent notions and tie each to the
component responsible for it. \emph{Document privacy} is the weakest: the
plaintext rotated database $E_{\mathrm{rot}}$ is visible to the server,
protected only by SVD truncation (through the proxy quantity $\sigma_{rec}$)
and the secret rotation $R$, an empirical obfuscation layer---this is
\emph{not} a cryptographic protection, and it fails under a known $R$ and, as
Section~\ref{ssec:alignment-pq} shows, under even a modest known-plaintext
alignment attack. \emph{Query-value confidentiality} comes from CKKS: the
numerical scoring vector and returned scores are encrypted, and our measured
SHARD path keeps the secret key client-side. It still reveals active cells,
ciphertext counts and candidate identifiers. \emph{Access-pattern
privacy} is again absent---the candidate IDs reranked per query, and the
public PQ codes, are not hidden, since CKKS leaves that channel untouched and
closing it would need composition with a PIR/ORAM-style primitive.

In particular, the public PQ artefact (codebook + per-document codes,
trained in the rotated space) is itself a lossy compressed view of
$E_{\mathrm{rot}}$. Section~\ref{ssec:alignment-pq} quantifies this
leakage on a controlled sample; the result is strong enough that public
PQ codes must be treated as a meaningful side channel rather than as a
harmless index artefact.

Table~\ref{tab:threat} summarises the coverage and
Table~\ref{tab:claims} maps each claim to the evidence that supports it.
The diagonal of the coverage table is intentional: ct-pt reranking takes
care of the value channel, but the access-pattern channel (which document
IDs are reranked) remains open, and the public-PQ-artefact channel is
measured as leaky rather than closed. They would require composition with
a PIR-style primitive and a stronger static-side index protection,
respectively.

\begin{table}[tbp]
\centering
\caption{Threat coverage of the global-linear \emph{baseline} (the foil).
\textsc{Shard}'s coverage of the same vectors is in
Table~\ref{tab:shard-claims}.}\label{tab:threat}
\small
\begin{tabular}{@{}p{0.42\linewidth}p{0.10\linewidth}p{0.42\linewidth}@{}}
\toprule
\textbf{Attack vector} & \textbf{Status} & \textbf{Comment} \\ \midrule
Snapshot of $E_{\mathrm{rot}}$ without $R$ (off-the-shelf Vec2Text) &
Partial &
Heuristic; empirically validated against non-adaptive Vec2Text. \\
Network interception, leakage of similarity scores &
Yes &
CKKS at tc128; $sk$ never leaves the client. \\
Known $R$ (known-rotation attacker) &
No &
Reduces to $\sigma_{rec}$ from SVD truncation only. \\
Known-plaintext pairs $(\mathrm{text}_i, E_{\mathrm{rot},i})$ &
No &
Evaluated in Section~\ref{ssec:alignment-pq}; roughly $k$ pairs
recover the rotation numerically. \\
Reference-corpus overlap after known-plaintext alignment &
Partial &
Evaluated in Section~\ref{ssec:reference-attack}; overlapping texts are
recoverable. \\
Unsupervised cross-space translation attacks &
No &
Not evaluated on the baseline; a covariance control is not treated as a full
translation attack. \\
Adaptive / learned alignment decoder &
Partial &
Experiment~29 applies the official pretrained GTR Vec2Text corrector to raw,
SHARD and partially aligned views. A SHARD-specific trained or universal
decoder remains unevaluated (Section~\ref{ssec:shard-text}). \\
Malicious server (protocol deviation, substitution) &
No &
Verifiable computation / TEEs are an orthogonal addition. \\
Access-pattern leakage on candidate IDs &
No &
CKKS hides values, not access patterns; requires PIR/ORAM. \\
Leakage from the public PQ artefact (rotated-space codes) &
Partial &
Evaluated in Section~\ref{ssec:alignment-pq}; neighbour structure is
substantially preserved, so this remains an exposed channel. \\
\bottomrule
\end{tabular}
\end{table}

\begin{table}[tbp]
\centering
\caption{\emph{Baseline} claims vs.\ evidence (the foil): which statements
about the global-linear stack are demonstrated and which are deliberately
\emph{not} claimed. \textsc{Shard}'s claims are in
Table~\ref{tab:shard-claims}.}
\label{tab:claims}
\small
\begin{tabular}{@{}p{0.40\linewidth}p{0.54\linewidth}@{}}
\toprule
\textbf{Claim} & \textbf{Evidence / status} \\ \midrule
CKKS hides query values and scores from the server &
Cryptographic construction; parameters from~\cite{ref_alpha}
(Sec.~\ref{ssec:crypto-detail}); query privacy is cryptographic. \\
ct-pt is faster than ct-ct for the reranking &
Established in~\cite{ref_alpha} ($1.7\times$ at the auto-tuned
configuration); reused unchanged. \\
Secret rotation helps against \emph{off-the-shelf}, non-adaptive
Vec2Text & Appendix~\ref{ssec:vec2text}; effect holds only for the
tested weak-attacker configuration. \\
Aligned off-the-shelf Vec2Text recovers text/PII after Procrustes &
Appendix~\ref{ssec:adaptive-inversion}: \textbf{not observed} in the
memory-capped RTX~4090 run, but the raw baseline is also near the
reconstruction floor, so this is not a proof of security. \\
Secret rotation is robust to known-plaintext alignment &
\textbf{False}; Section~\ref{ssec:alignment-pq} shows near-exact
recovery with about $k$ known pairs. \\
Public PQ codes are harmless metadata &
\textbf{False}; Section~\ref{ssec:alignment-pq} shows substantial
reconstruction and nearest-neighbour leakage. \\
The protective wrapper preserves ranking in $\mathrm{span}(V_k)$ &
Metric-preserving, Pearson $>0.9999$~\cite{ref_alpha} (recap Sec.~\ref{ssec:projection-baselines}). \\
SVD truncation improves \emph{Acc@1} (a ``denoiser'') &
\textbf{Not significant}~\cite{ref_alpha}: paired tests give
$p\geq0.11$ for every encoder. \\
SVD truncation significantly improves \emph{Acc@10} for $d{=}1024$
retrieval-trained encoders &
\textbf{True}~\cite{ref_alpha} (e5-large $+0.042$, $p<10^{-6}$); it
significantly \emph{degrades} compact/paraphrase encoders. \\
The denoiser generalises to real IR &
\textbf{False}~\cite{ref_alpha}: SVD significantly \emph{reduces}
nDCG@10 on BEIR SciFact/NFCorpus for every encoder. \\
$\sigma_{rec}$ alone measures document privacy &
\textbf{Not claimed}; Appendix~\ref{app:noise} shows that
matched-distortion noise can preserve substantially more raw neighbour
structure. \\
End-to-end query latency $<1$~s at $10^6$ docs &
\textbf{True}~\cite{ref_alpha}; recap Sec.~\ref{ssec:latency} (loopback PoC). \\
The system is secure against an \emph{adaptive} inversion attacker &
\textbf{Not claimed}; learned/universal adaptive decoders remain a
required experiment (Sec.~\ref{sec:limits}). \\
The system is secure against reference-corpus lookup &
\textbf{False when the reference overlaps}; Section~\ref{ssec:reference-attack}
shows exact paragraph recovery after alignment. \\
The system is secure against unsupervised cross-space transfer attacks &
\textbf{Not claimed}; these attacks remain outside the evaluated threat
model. \\
The rotation is a cryptographic guarantee &
\textbf{Not claimed}; it is an empirical obfuscation layer only. \\
Document privacy on the server is cryptographic &
\textbf{Not claimed}; $E_{\mathrm{rot}}$ is plaintext on the server. \\
No information leaks from the public PQ codes &
\textbf{Not claimed}; the measured leakage is non-negligible. \\
\bottomrule
\end{tabular}
\end{table}

\begin{table}[tbp]
\centering
\caption{\textsc{Shard} (the contribution) claims vs.\ evidence.}
\label{tab:shard-claims}
\footnotesize
\begin{tabular}{@{}p{0.42\linewidth}p{0.52\linewidth}@{}}
\toprule
\textbf{Claim} & \textbf{Evidence / status} \\ \midrule
Corrected full-dimensional scoring preserves raw ranking &
\textbf{Supported algebraically and empirically}; Experiment~23 obtains an
identical order in every sampled float64 check across ten BEIR/MIRACL cells.
End-to-end quality still depends on shortlist recall
(Section~\ref{ssec:shard-utility}). \\
Half-PCA alone causes the formerly reported $2$--$8$ point loss &
\textbf{False}; query centring caused part of it. Corrected truncation gaps
range from approximately $0$ to $5.3$ nDCG points. \\
Cell keys create a hard $d_{\mathrm{priv}}$ recovery threshold &
\textbf{False}; minimum-norm OLS exceeds $0.9$ residual-gallery R@1 with
about $32$--$36$ in-cell pairs while the full key remains underdetermined
(Section~\ref{ssec:shard-align}). \\
An unknown cell key suppresses off-the-shelf text reconstruction &
\textbf{Supported only in the strengthened GTR case}; mean token-F1 falls
from $.665$ to $.242$ without anchors, but the public prefix and eight-pair OLS
restore substantial output. This is not a general inversion guarantee
(Section~\ref{ssec:shard-text}). \\
Increasing $C$ compartmentalises diffuse known pairs &
\textbf{Supported in the measured setting}; the global OLS budget for
R@1$\ge0.9$ shifts from $32$ to $8\,192$ over $C=1$ to $256$. The associated
active-cell and packed-CKKS costs are measured separately; this is not a
formal security threshold. \\
The public prefix leaks less than the baseline's index &
Section~\ref{ssec:shard-leak}; NN-overlap@10 $0.20$--$0.55$ vs.\ $0.76$. \\
Micro-keys remove shared residual edges and are unlinkable &
The first clause is supported; \textbf{unlinkability is false}. Norm-rank and
the unchanged prefix link persistent rows almost perfectly even with
insert/delete churn (Section~\ref{ssec:shard-micro}). \\
The earlier Gaussian perturbation was a DP baseline &
\textbf{False}; it lacked adjacency, clipping and accounting. Experiment~27
replaces it with an analytically calibrated $(\varepsilon,\delta)$ Gaussian
release under explicit replacement adjacency
(Section~\ref{ssec:shard-vs-dp}). \\
SHARD's in-process cryptographic path is measured &
\textbf{Supported for the local measured path}; 315 real TenSEAL/SEAL runs
cover encryption, serialization, ct--pt evaluation and decryption. Network
RTT, concurrency and packed multi-score responses remain unmeasured
(Section~\ref{ssec:shard-ckks}). \\
\textsc{Shard} itself is differentially private &
\textbf{False}; Experiment~27 is a separately evaluated Gaussian-release
baseline. No $(\varepsilon,\delta)$ guarantee is claimed for SHARD. \\
\textsc{Shard} hides the neighbour graph &
\textbf{Not claimed}; same-cell similarities leak
(Section~\ref{ssec:shard-leak}). \\
\textsc{Shard} stops an overlapping reference lookup &
\textbf{False}; the public prefix still de-anonymises
(Section~\ref{ssec:shard-ref}). \\
\textsc{Shard} is a cryptographic guarantee &
\textbf{Not claimed}; it is a geometric compartmentalisation mechanism with
explicit leakage. \\
\bottomrule
\end{tabular}
\end{table}

\section{Baseline Construction (the global-linear foil)}\label{sec:method}

This section describes the global-linear \emph{baseline} that
\textsc{Shard} (Section~\ref{sec:shard}) replaces---SVD truncation, a single
secret rotation, a public PQ index, and CKKS reranking. The construction is
\emph{reproduced} from our prior global-linear work~\cite{ref_alpha} and condensed
to the elements the \textsc{Shard} comparison needs; we claim no novelty for
it here and refer the reader there for the full derivation. It is the foil
against which the contribution is measured; \textsc{Shard} reuses its CKKS
ct-pt machinery (Section~\ref{ssec:online}) but
discards its single global geometry. The baseline is illustrated in
Figure~\ref{fig:architecture}: an offline data-preparation phase
(Section~\ref{ssec:offline}) and the online query protocol
(Section~\ref{ssec:online}).

\begin{figure}[tbp]
\centering
\begin{tikzpicture}[
  font=\small,
  box/.style={draw=black!55, rounded corners=2pt, align=center, font=\footnotesize,
              minimum height=9mm, text width=22mm, fill=blue!5, inner sep=2pt},
  >={Stealth[length=2mm]}, line width=0.5pt]
  \node[box] (enc) {1.~Client\\encoder};
  \node[box, right=6mm of enc] (tr) {2.~Secret transform\\$\mu, V_k, R$};
  \node[box, right=6mm of tr] (pq) {3.~Local PQ\\shortlist};
  \node[box, right=6mm of pq] (ck) {4.~CKKS encrypt\\rotated query};
  \node[box, below=12mm of ck, fill=red!6] (srv) {5.~Server\\ct-pt rerank};
  \node[box, left=6mm of srv] (dec) {6.~Decrypt\\scores};
  \node[box, left=6mm of dec] (top) {7.~Top-$k$\\results};
  \draw[->] (enc) -- (tr);
  \draw[->] (tr) -- (pq);
  \draw[->] (pq) -- (ck);
  \draw[->] (ck) -- (srv);
  \draw[->] (srv) -- (dec);
  \draw[->] (dec) -- (top);
  \node[font=\footnotesize\itshape, text=black!70, anchor=south west]
    at ([yshift=5mm]enc.north west) {Offline: protected database $E_{\mathrm{rot}}$ + public PQ artifact};
  \node[font=\footnotesize\itshape, text=black!70, anchor=north west]
    at ([yshift=-5mm]top.south west) {Online: CKKS hides query values \& scores; access pattern stays visible};
\end{tikzpicture}
\caption{Online query flow of the global-linear \emph{baseline}.
Steps 1--7 realise the two-stage protocol: client transformation
$T(\cdot)$ and local PQ filtering produce a short-list of
$K_{\mathrm{cands}}$ candidate IDs; step~4 encrypts the rotated query
under CKKS; step~5
runs ct-pt reranking on the server; steps~6--7 decrypt the scores and
sort. The secret keys ($\mu$, $V_k$, $R$, $sk_{\mathrm{CKKS}}$) and the
rotated database $E_{\mathrm{rot}}$ are produced offline by the data
owner (Section~\ref{ssec:offline}).}
\label{fig:architecture}
\end{figure}

\subsection{Offline phase}\label{ssec:offline}

\paragraph{Encoder and centring.} The data owner runs the encoder $E$ on a
corpus $D = \{T_1, \dots, T_{N_{\mathrm{docs}}}\}$ to obtain
$X \in \mathbb{R}^{N_{\mathrm{docs}} \times d}$ with $L_2$-normalised rows.
The global centroid $\mu \in \mathbb{R}^d$ is computed and the corpus is
centred: $X_{\mathrm{c}} = X - \mu$.

\paragraph{SVD truncation.} A randomised SVD~\cite{ref_halko,ref_eckart}
$X_{\mathrm{c}} \approx U_k \Sigma_k V_k^\top$ produces the orthonormal
matrix $V_k \in \mathbb{R}^{d \times k}$ that spans the dominant
$k$-dimensional subspace. We define the corpus-level relative
reconstruction error
\[
\sigma_{rec}^{2}(E; V_k) =
\frac{\|E - \pi_k(E)\|_F^2}{\|E\|_F^2} = 1 - \eta_k ,
\]
where $\eta_k$ is the fraction of squared Frobenius energy retained.
Because $\sigma_{rec}$ \emph{decreases} as $k$ grows, the operating point
$k = d/2$ is the \emph{largest} truncation dimension (rounded to $d/2$) at
which the distortion floor $\sigma_{rec} \geq 0.10$ still holds uniformly
across all five tested encoders (the binding cases are e5-large,
$\sigma_{rec}=0.101$, and mpnet, $\sigma_{rec}=0.107$); larger $k$ would
drop some encoders below the floor. Thus $k=d/2$ maximises retained
utility subject to the distortion-proxy constraint. We caution that the
threshold $0.10$ is an engineering proxy (Section~\ref{sec:lemma}), not a
privacy guarantee, and is close to binding for two encoders.

\paragraph{Secret rotation.} A Haar-uniform random orthogonal matrix
$R \in O(k)$ is generated by QR-decomposing a Gaussian matrix and
correcting the diagonal of $R$ to obtain a uniform distribution on the
group. The protected document representation is
\[
v'_i \;=\; T(v_i) \;=\; R\, V_k^\top (v_i - \mu) \in \mathbb{R}^k .
\]

\paragraph{Public PQ artefact.} A faiss product-quantisation
\texttt{IndexPQ} index~\cite{ref_pq} is trained
\emph{in the rotated space} $E_{\mathrm{rot}}$ with $M = k/4$ subquantisers
and $8$ bits per code. Both the codebook and the per-document codes are
public; the client downloads them once at on-boarding. Training the PQ
artefact in the rotated space removes one trivial cross-space pairing
(an attacker cannot compose pairs
$(\hat E_{\mathrm{proj}}, E_{\mathrm{rot}})$ from a non-rotated PQ index
to estimate $R$). It does \emph{not}, however, make the PQ codes safe:
the public codes are a lossy quantised image of $E_{\mathrm{rot}}$ and
may themselves leak neighbourhood, cluster or topic structure. We treat
the PQ artefact as part of the attack surface and quantify
nearest-neighbour and reconstruction leakage in
Section~\ref{ssec:alignment-pq}.

\paragraph{Client-side state.} The client retains $\mu$, $V_k$, $R$ and
the CKKS secret key $sk_{\mathrm{CKKS}}$. The server stores
$E_{\mathrm{rot}} \in \mathbb{R}^{N_{\mathrm{docs}} \times k}$ in plaintext
and the public PQ artefact.

\subsection{Online query protocol}\label{ssec:online}

The pipeline of Figure~\ref{fig:architecture} processes a query in stages, of
which only the last leaves the client. The client encodes the query,
$v_q = E(q_{\mathrm{text}}) \in \mathbb{R}^d$, and forms two projected
coordinates: the centred router
$\widetilde q=R V_k^\top(v_q-\mu)$ and the rank-correct scoring query
$q_{\mathrm{score}}=R V_k^\top v_q$. Using the public PQ artefact and
$\widetilde q$ it runs an asymmetric PQ-distance search locally and
keeps the top-$K_{\mathrm{cands}}$ candidate IDs $I_{\mathrm{cand}}$---a
short-list small enough ($K_{\mathrm{cands}} = 40$ in our experiments) that the
CKKS reranking still fits the latency budget. It encrypts the scoring query,
$\mathrm{ct}_q = \mathrm{Enc}_{pk}(q_{\mathrm{score}})$, and sends it on; for each
$j \in I_{\mathrm{cand}}$ the server computes
$\mathrm{ct}_{\mathrm{score}}^{(j)} = \mathrm{ct}_q \odot v'_j$ in ct-pt mode,
a ciphertext that encodes
\[
\langle q_{\mathrm{score}},v'_j\rangle
=v_q^\top V_kV_k^\top(v_j-\mu),
\]
the projected raw score up to a document-independent query offset. The client finally decrypts those scores,
$s_j = \mathrm{Dec}_{sk}(\mathrm{ct}_{\mathrm{score}}^{(j)})$, sorts them, and
returns the top-$K$ documents---so the server reranks without ever seeing a
plaintext query value or score.

\paragraph{Why ct-pt is sufficient.} In step~5 the server multiplies a
ciphertext by a plaintext vector. Unlike ciphertext-ciphertext
multiplication, this operation produces a ciphertext that stays in the
two-component form, no relinearisation is required, and the modulus chain
is consumed by exactly one rescale. The sum of slots into a single
scalar $\langle q_{\mathrm{score}}, v'_j \rangle$ requires $\log_2 k$ rotations using
Galois keys, but no further multiplications. The result is a sub-second
latency at $N_{\mathrm{docs}} = 10^6$.

\paragraph{Choice of $K_{\mathrm{cands}}$.} Since each of the
$K_{\mathrm{cands}}$ candidates is processed by one ct-pt operation, the
server-side latency grows linearly in $K_{\mathrm{cands}}$. Conversely, a
small $K_{\mathrm{cands}}$ trades against PQ-recall: the relevant
document must be in the short-list for the CKKS reranker to score it.
We pick $K_{\mathrm{cands}}$ as the smallest value for which
PQ-recall@$K_{\mathrm{cands}}$ matches the SVD-projected exact baseline
within the 5-seed CI; in our integral experiment this is
$K_{\mathrm{cands}} = 40$, but the parameter is a deployment knob and
should be retuned for new encoders or larger corpora.

\subsection{Cryptographic reproducibility}\label{ssec:crypto-detail}

For the auto-selected configuration
($N_{\mathrm{poly}} = 8192$, coefficient-modulus chain $[60,40,60]$ bits,
$\log_2 Q = 160$, scale $\Delta = 2^{40}$, TenSEAL/Microsoft SEAL
backend)---selected by the reproducible offline grid search of~\cite{ref_alpha}
($\approx 1.7\times$ faster than the TenSEAL stock $[60,40,40,60]$ chain at equal
security), reused here unchanged---we report the parameters needed to reproduce the cryptographic
budget rather than only the latency:

A single scoring query $q_{\mathrm{score}} \in \mathbb{R}^{k}$ packs into one ciphertext over
$k \leq N_{\mathrm{poly}}/2 = 4096$ slots, and since the integral encoders have
$k \in \{192,384,512\}$ one ciphertext per query suffices, with no
cross-ciphertext aggregation. Each candidate score then costs a single ct-pt
multiply, one rescale (one level consumed), and a slot-sum realised as
$\lceil \log_2 k \rceil$ ciphertext rotations under Galois keys ($8$--$9$
rotations for the tested $k$)---never a relinearisation or a
ciphertext-ciphertext multiply. The $[60,40,60]$ chain leaves exactly one
multiplicative level beyond the input, and after the lone rescale the
remaining $60$-bit modulus holds the additive noise far below the
$\Delta = 2^{40}$ scale, so decrypted scores track the plaintext inner product
to a relative error under $10^{-3}$ (Pearson $>0.9999$ against exact,
Section~\ref{ssec:integral}). The objects stay small: a fresh ciphertext is
$\approx 0.21$~MB, the Galois-key set for the power-of-two rotations is
$\approx 6$--$8$~MB, the relinearisation key is \emph{not} generated at all
(ct-pt only), and the public key is $\approx 0.4$~MB, so per query the client
uploads one ciphertext ($\approx 0.21$~MB) and downloads the
$K_{\mathrm{cands}}$ score ciphertexts ($\approx 40 \times 0.21 \approx 8.4$~MB
before base64; measured on-wire JSON+base64 sizes are in
Section~\ref{ssec:latency}). On security we assert no generic ``standard''
label, reporting instead $\log_2 Q = 160 < 218$ for $N_{\mathrm{poly}}=8192$,
within the conservative ternary-secret \texttt{tc128} bound of the
HomomorphicEncryption.org tables \cite{ref_he_std} and cross-checked against
the lattice-estimator methodology of Albrecht et al.\
\cite{ref_lwe_estimator}; the standard documents are guidelines, so we name the
concrete estimator and table used so the claim can be re-derived.

\section{A Projection Bound for the Distortion Proxy}\label{sec:lemma}

This short section concerns the \emph{baseline} only (the bound is
reproduced from~\cite{ref_alpha}): it states the one elementary fact behind
its $\sigma_{rec}$ distortion proxy, and why that proxy is not a privacy
metric---motivating \textsc{Shard}'s move to an
attack-aware design that does not rely on $\sigma_{rec}$ at all. Let
$\mathbf{x} \in \mathbb{R}^d$ denote a centred embedding,
$V_k \in \mathbb{R}^{d \times k}$ the orthonormal matrix from the SVD
truncation, and $\pi_k(\mathbf{x}) = V_k V_k^\top \mathbf{x}$ the
orthogonal projection onto $\mathrm{span}(V_k)$. We use a \emph{per-vector}
relative reconstruction error here,
$\sigma_{rec}(\mathbf{x}; V_k) =
\| \mathbf{x} - \pi_k(\mathbf{x}) \|_2 / \|\mathbf{x}\|_2$,
distinguished from the corpus-level Frobenius quantity
$\sigma_{rec}(E;V_k)$ of Section~\ref{ssec:offline} (the operating-point
threshold $0.10$ is applied to the corpus-level aggregate); write
$\mathbf{x}_\perp = \mathbf{x} - \pi_k(\mathbf{x})$.

\begin{proposition}\label{lemma:proj-l2}
Let $f : \mathbb{R}^d \to \mathbb{R}^d$ be any decoder whose image is
constrained by $f(\mathbf{y}) \in \mathrm{span}(V_k)$ for every
$\mathbf{y} \in \mathbb{R}^d$. Then for every
$\mathbf{x} \in \mathbb{R}^d$
\[
\| \mathbf{x} - f(\pi_k(\mathbf{x})) \|_2^{2}
\;\geq\; \| \mathbf{x}_\perp \|_2^{2},
\]
with equality at $f(\mathbf{y}) = \mathbf{y}$. In relative form,
$\| \mathbf{x} - f(\pi_k(\mathbf{x})) \|_2 / \|\mathbf{x}\|_2
\geq \sigma_{rec}(\mathbf{x}; V_k)$.
\end{proposition}

\begin{proof}
Decompose $\mathbf{x} = \pi_k(\mathbf{x}) + \mathbf{x}_\perp$ with
$\pi_k(\mathbf{x}) \in \mathrm{span}(V_k)$ and
$\mathbf{x}_\perp \in \mathrm{span}(V_k)^\perp$. By the image
restriction $f(\pi_k(\mathbf{x})) \in \mathrm{span}(V_k)$, hence
$\mathbf{x} - f(\pi_k(\mathbf{x})) =
\bigl( \pi_k(\mathbf{x}) - f(\pi_k(\mathbf{x})) \bigr) + \mathbf{x}_\perp$
is the sum of two orthogonal vectors. Pythagoras gives
$\| \mathbf{x} - f(\pi_k(\mathbf{x})) \|_2^{2}
= \| \pi_k(\mathbf{x}) - f(\pi_k(\mathbf{x})) \|_2^{2}
+ \| \mathbf{x}_\perp \|_2^{2}
\geq \| \mathbf{x}_\perp \|_2^{2}$.
\end{proof}

\paragraph{Scope: a projection bound, not an inversion-security
theorem.} Proposition~\ref{lemma:proj-l2} bounds the $L_2$ error of
decoders \emph{whose image is contained in} $\mathrm{span}(V_k)$. It is a
statement about information lost to the projection, not about the
security of text inversion. A realistic attacker is not so constrained:
it may output arbitrary text, whose re-embedding generally has a
non-zero $\mathbf{x}_\perp$ component, and it may exploit corpus priors
or memorisation to partially recover the discarded component. The
proposition therefore does \emph{not} lower-bound the achievable inversion BLEU,
token overlap or PII recovery. We use it only to make the
\emph{proxy} criterion $\sigma_{rec} \geq 0.10$ precise within the
projection-restricted class, and we stress that the threshold $0.10$ is
an \emph{engineering proxy}, calibrated on one encoder (GTR-base) against
one off-the-shelf attacker, \emph{not} a security threshold. A
per-encoder ablation of $\sigma_{rec}$ against BLEU, token overlap and
typed-PII recovery (names, addresses, e-mail, phone, medical terms) is
listed as a required pre-submission experiment
(Section~\ref{sec:limits}).

\begin{hypothesis}\label{hyp:bleu}
For an off-the-shelf inversion attack \textsc{Vec2Text}~\cite{ref_vec2text},
the expected BLEU of recovering the original text from
$\pi_k(\mathbf{x})$ is monotone in
$\eta_k = 1 - \sigma_{rec}^{2}(E; V_k)$:
\[
\mathbb{E}\bigl[ \mathrm{BLEU}(f(\pi_k(\mathbf{x})), T(\mathbf{x})) \bigr]
\;\approx\; \mathrm{BLEU}_0 + \gamma_f \, \eta_k,
\]
where $\mathrm{BLEU}_0$ is the BLEU of a ``random semantically close''
text and $\gamma_f$ is a decoder-specific constant.
\end{hypothesis}

Hypothesis~\ref{hyp:bleu} is examined numerically in
Appendix~\ref{ssec:vec2text}; we deliberately separate the formal
statement from the empirical observation.

\section{SHARD: Cell-Keyed Residual Splitting}\label{sec:shard}

The baseline of Section~\ref{sec:method} protects the collection with a
single global geometry $E_{\mathrm{rot}} = R\,V_k^\top(x-\mu)$. Its
weakness is structural: one global map is recoverable by one alignment
(Section~\ref{ssec:alignment-pq}), and the public index sits on the full
protected space. \textsc{Shard} keeps the same CKKS ct-pt reranking but
removes the single global axis. It is a \emph{family} of transforms
parametrised by a cell count $C$.

\paragraph{Offline construction.} Let $\mu$ be the corpus centroid and
$V\in O(d)$ the PCA basis fitted to centred document embeddings (eigenvectors
of the covariance, sorted by variance). For a document $x_i$ write
\[
V^\top(x_i-\mu)=[\,u_i\,\|\,r_i\,],
\]
splitting its centred PCA coordinate into a
\emph{public prefix} $u\in\mathbb{R}^{d_{\mathrm{pub}}}$ (the top variance
directions) and a \emph{private residual} $r\in\mathbb{R}^{d_{\mathrm{priv}}}$,
$d_{\mathrm{priv}}=d-d_{\mathrm{pub}}$. A coarse partition into $C$
\emph{cells} is defined by $k$-means on $u$. The data owner assigns one
orthogonal key $H_c\in O(d_{\mathrm{priv}})$ per cell. The experiments sample
dense Haar-like matrices by Gaussian QR under deterministic research seeds;
a deployment would need a CSPRNG-derived key schedule and should benchmark a
structured fast-orthogonal representation. The stored
representation of document $i$ in cell $c(i)$ is the pair
\[
\big(\, u_i,\ \ z_i = H_{c(i)}\, r_i \,\big),
\]
the prefix $u_i$ for the stage-1 index and the keyed residual shard $z_i$
for reranking. The keys $\{H_c\}$, $V$ and $\mu$ are client secrets.

\paragraph{Online two-stage query.} Routing and final scoring use two related
query coordinates. The centred prefix
$\widetilde u_q=V_{\mathrm{pub}}^\top(q-\mu)$ drives stage-1 ANN, while final
scoring uses
\[
u_q=V_{\mathrm{pub}}^\top q,\qquad
r_q=V_{\mathrm{priv}}^\top q.
\]
The two prefix vectors differ only by the fixed offset
$V_{\mathrm{pub}}^\top\mu$. The client (i)~runs stage-1 ANN locally with
$\widetilde u_q$ to obtain $K_{\mathrm{cands}}$ candidates; then (ii)~for
every active cell $c$ forms $H_c r_q$. In the width-1 layout it sends one
$\mathrm{Enc}(H_c r_q)$ per cell; block-SIMD optionally concatenates $B$
cell vectors in one ciphertext, with zero-masked plaintext operands selecting
the candidate's block. The server returns ct--pt inner products against the
stored shards $z_i$. Because $H_c$
is orthogonal,
\[
\langle H_c r_q,\ z_i\rangle=\langle H_c r_q,\ H_c r_i\rangle=\langle r_q, r_i\rangle ,
\]
in exact plaintext arithmetic. CKKS adds its usual numerical approximation.
The scoring-prefix term can be added as plaintext to the encrypted response
before decryption, as in Experiment~26, or combined by the client afterwards:
\[
S(q,i)=\langle u_q,u_i\rangle+\langle r_q,r_i\rangle
=q^\top(x_i-\mu)=q^\top x_i-q^\top\mu .
\]
The last term is document-independent, so the exact-arithmetic ordering is the
raw dot-product ordering. In the plaintext evaluation, the remaining measured
retrieval loss is whether stage~1 places a relevant document in the short-list;
Experiment~26 performs the complementary check with actual CKKS: across 315
trials it finds no top-1 flips, unit top-10 overlap and maximum absolute score
error $2.29\times10^{-6}$.

\begin{proposition}[Rank-preserving cell-keyed reranking]\label{prop:exact}
Let $H_c\in O(d_{\mathrm{priv}})$ be the orthogonal key of cell $c$ and
$z_i=H_{c(i)}r_i$ the stored shard of document $i$ in cell $c(i)$. For any
query residual $r_q$ and any document $i$,
\[
\big\langle H_{c(i)}\,r_q,\ z_i\big\rangle=\langle r_q,r_i\rangle ,
\]
so $S(q,i)=q^\top x_i-q^\top\mu$ is reconstructed with no truncation in
exact arithmetic. Consequently,
$S(q,i)-S(q,j)=q^\top x_i-q^\top x_j$, and the raw document ranking is
preserved. CKKS evaluation is approximate, not algebraically exact.
Consequently, for documents $i,j$ in the \emph{same} cell,
$\langle z_i,z_j\rangle=\langle r_i,r_j\rangle$.
\end{proposition}
\begin{proof}
$H_c$ is orthogonal, hence $H_c^\top H_c=I$ and
$\langle H_c r_q,H_c r_i\rangle=r_q^\top H_c^\top H_c\,r_i=\langle r_q,r_i\rangle$;
the prefix and residual terms together give $q^\top(x_i-\mu)$. Its offset
from $q^\top x_i$ is constant across documents. The same-cell identity
applies the same equality to the shared key $H_c$ of two stored residuals.
\end{proof}

\paragraph{Cross-release invariants and a non-claim.} Orthogonal keying is
invertible for a party that holds the key; without the key it is an
unidentified orientation, not a non-invertible template. Fresh orientations
also do not make two snapshots unlinkable. For two releases $a,b$ of the same
residual under independent orthogonal keys,
\[
\|z_i^{(a)}\|_2=\|r_i\|_2=\|z_i^{(b)}\|_2,
\qquad
\langle z_i^{(a)},z_j^{(a)}\rangle
=\langle r_i,r_j\rangle
=\langle z_i^{(b)},z_j^{(b)}\rangle
\]
whenever $i$ and $j$ share a cell. The prefix $u_i$ is unchanged as well.
Norms, within-cell Gram signatures and the prefix are therefore stable
cross-release identifiers. Re-keying changes coordinates and preserves
plaintext retrieval, but \textsc{Shard} does not claim unlinkability,
irreversibility or a cancellable-template guarantee. We measure these
channels explicitly in Section~\ref{ssec:shard-micro}.

\paragraph{Proved vs.\ measured, and what \textsc{Shard} does not promise.}
What \textsc{Shard} \emph{proves} is precisely Proposition~\ref{prop:exact}:
raw-rank-equivalent full-dimensional reranking in plaintext algebra, and the same-cell cancellation
$\langle z_i,z_j\rangle=\langle r_i,r_j\rangle$. Everything else below is
\emph{measured}, not guaranteed: partial known-pair alignment
(\S\ref{ssec:shard-align}), learned and unsupervised attackers
(\S\ref{ssec:shard-learned}), public-prefix leakage
(\S\ref{ssec:shard-leak}), micro-key and cross-release behaviour
(\S\ref{ssec:shard-micro}), and online cost
(\S\ref{ssec:shard-cost}). Three non-promises follow directly from the
proposition and the construction: \textsc{Shard} is \emph{not} a cryptographic
guarantee on the store; it does \emph{not} hide the coarse neighbour graph
(same-cell similarities are computable by the server unless $C{=}N$); and an
\emph{overlapping plaintext reference corpus} still de-anonymises through the
public prefix (\S\ref{ssec:shard-ref}). The effect \textsc{Shard} measures is
compartmentalisation of a \emph{diffuse} alignment stream, not a hard privacy
barrier or protection of the coarse graph.

\paragraph{The $C$ family and micro-keys.} The cell count interpolates
between two regimes. At $C{=}1$ the residual carries a single global key;
this is a full-dimensional SHARD control, not the earlier truncated baseline.
At the other extreme, a \emph{micro-key} variant assigns one key per
document ($C{=}N$). No two residuals then share a key, so the directly
comparable within-release residual graph vanishes, at the cost of one
transformed query per candidate. This does not remove the stable prefix or
residual norm and therefore does not imply cross-release unlinkability.

\paragraph{What the server sees, and what it protects.} The honest-but-curious
server holds the coarse prefix $u$, cell labels and the cell-keyed shards $z$;
it also observes candidate identifiers, ciphertext counts and the query access
pattern. Three facts
frame every experiment below. (1)~\emph{Within a cell the keys cancel}:
$\langle z_i,z_j\rangle=\langle r_i,r_j\rangle$ for $i,j$ in the same cell,
so same-cell similarities are computable by the server; \textsc{Shard}
localises but does not eliminate neighbour-graph leakage. (2)~\emph{Partial
mapping precedes exact key recovery}: known pairs identify the key action on
their span even while the complete $H_c$ remains underdetermined. Exact key
identification generically needs full in-cell rank, but retrieval and record
matching may succeed much earlier. (3)~The stable prefix, residual norms and
within-cell Gram signatures provide direct cross-release linkage channels.
Cell sharding reduces the rate at which a diffuse anchor stream fills any one
cell; it does not create a hard privacy threshold.

\section{Experiments}\label{sec:experiments}

The experiments come in two groups. Sections~\ref{ssec:alignment-pq}--\ref{ssec:latency} characterise the
global-linear baseline, opening with the failure modes that motivate
\textsc{Shard}: known-plaintext Procrustes alignment recovers the rotation
from $\approx k$ pairs, the public PQ codes leak neighbour structure, and an
overlapping reference corpus turns alignment into exact lookup
(Section~\ref{ssec:alignment-pq}); the baseline's utility and latency are
then recapped from~\cite{ref_alpha}. The
inversion-versus-rotation analyses reproduced from the global-linear
paper~\cite{ref_alpha}---off-the-shelf and aligned Vec2Text, the
lightweight projection baselines, and the calibrated-noise diagnostic
that $\sigma_{rec}$ is not a privacy metric---are condensed in
Appendix~\ref{app:baseline}. Sections~\ref{ssec:shard-utility}--\ref{ssec:shard-ref}
then \emph{evaluate \textsc{Shard}}: utility, the online cost, alignment
under rank-deficient and minimum-norm estimators, public-index leakage,
cross-release linkage for cell and micro-keys, and the overlap-reference
limitation. The maximal audit then adds a real CKKS/block-SIMD execution,
insert/delete churn, a formally calibrated Gaussian release, and a learned
text-inversion outcome. The former uncalibrated ``DP-noise'' comparison is
retained only as an artifact-side distortion diagnostic.
Readers interested only in the contribution can skip to
Section~\ref{ssec:shard-utility}; the baseline group is the foil.

\paragraph{Hardware and software.} All local experiments run on a single
workstation: Intel Core i5-14400F, 32~GB DDR5-4800, NVIDIA RTX~5060
(8~GB VRAM), Windows 11. The baseline artifact retains its original pinned
environment and CKKS build information. The corrective Experiments~23--25
reported here ran under Python~3.12.3, NumPy~2.4.6, SciPy~1.17.1,
scikit-learn~1.8.0 and Matplotlib~3.10.9; each output directory records the
command, seeds, configuration and Git revision. Experiments~26--28 use the
same workstation. Experiment~26 adds TenSEAL~0.3.16 and psutil~7.2.2, pins one
logical thread per P-core, and executes SEAL CKKS on the CPU; Experiments~27
and~28 use the recorded NumPy/SciPy stack. The local learned-inversion run
uses the RTX~5060 and its own frozen package report. The earlier aligned
Vec2Text stress test retained in Appendix~\ref{ssec:adaptive-inversion} used
a separate memory-capped RTX~4090 profile with roughly 11~GiB free VRAM,
Python~3.12 and PyTorch~2.6.0+cu124. We keep these environments distinct
rather than treating latency across machines as comparable.

\subsection{Attacks on the baseline: alignment, PQ leakage, and inversion}\label{ssec:alignment-pq}

We now stress the global-linear baseline with four attacks that step outside
the primary non-adaptive threat model: known-plaintext alignment of the
secret rotation, leakage from the public PQ codes, a text-level
reference-corpus lookup, and an aligned off-the-shelf inverter.

\paragraph{Known-plaintext alignment.} We simulate an attacker who knows
$m$ pairs $(E(x_i), E_{\mathrm{rot},i})$ and estimates the hidden
orientation with the orthogonal Procrustes solution
\cite{ref_schonemann}. The experiment uses e5-small projected to
$k=192$, five rotation seeds, $400$ held-out probes and a
$10\,000$-document gallery. The attack is evaluated in embedding space:
after alignment, the probe should recover its own projected vector and
retrieve the corresponding document from the gallery. This does not
measure generated text quality, but it is the relevant precursor to a
few-shot Vec2Text-style attack because it recovers the native projected
space given to the inverter.

\begin{table}[tbp]
\centering
\caption{Known-plaintext Procrustes attack. $m=-1$ is the known-$R$
oracle; $m=0$ is no alignment. Values are means over five rotation
seeds.}
\label{tab:procrustes}
\small
\begin{tabular}{@{}rrrrr@{}}
\toprule
Known pairs $m$ & Mean cosine & Rel.\ $L_2$ error & Target R@1 & Target R@10 \\ \midrule
$0$   & $-0.004$ & $1.417$ & $0.000$ & $0.001$ \\
$10$  & $0.098$  & $1.343$ & $0.005$ & $0.039$ \\
$50$  & $0.377$  & $1.116$ & $0.722$ & $0.925$ \\
$100$ & $0.621$  & $0.870$ & $1.000$ & $1.000$ \\
$192$ & $1.000$  & $1.8{\times}10^{-5}$ & $1.000$ & $1.000$ \\
$500$ & $1.000$  & $4.9{\times}10^{-7}$ & $1.000$ & $1.000$ \\
\bottomrule
\end{tabular}
\end{table}

The conclusion is sharp: the rotation is useful only while aligned
plaintext anchors are unavailable. With approximately $k$ known pairs,
the orientation is numerically recovered. Even $50$ pairs do not recover
the vector perfectly, but they are enough to retrieve the target in the
top~10 for $92.5\%$ of probes in a $10\,000$-document gallery. This is a
negative result for broad document-privacy claims and is therefore
reflected in the threat table.

\begin{figure}[tbp]
\centering
\includegraphics[width=0.58\linewidth]{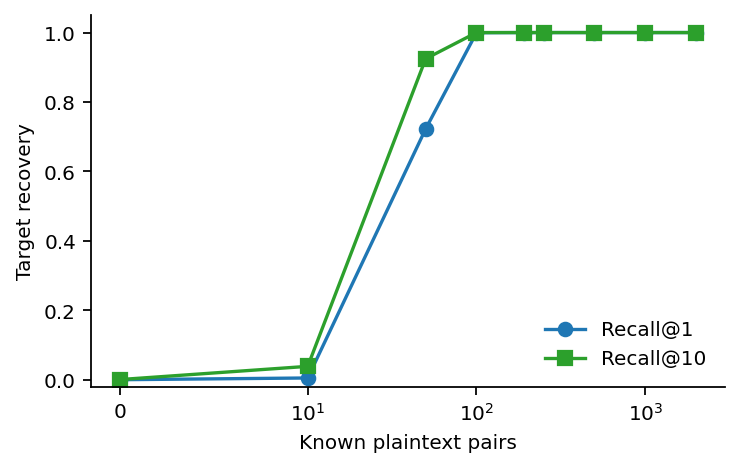}
\caption{Known-plaintext alignment quickly destroys the unknown-rotation
assumption. The $x$ axis uses a symlog scale to include $m=0$.}
\label{fig:procrustes}
\end{figure}

\paragraph{PQ-code leakage.} We also train public product-quantisation
codebooks directly in the rotated space and reconstruct approximate
vectors from codebook + per-document codes alone. On a $20\,000$-document
sample, the canonical e5-small configuration ($M=48$, 8 bits) preserves
a mean cosine of $0.953$ to $E_{\mathrm{rot}}$, $67.4\%$ of exact top-10
neighbours, and $97.1\%$ of exact top-10 neighbours inside the approximate
top-40 candidate set. Smaller artefacts leak less but remain
non-negligible (Table~\ref{tab:pq-leakage}).

\begin{table}[tbp]
\centering
\caption{Leakage from public PQ codes on a $20\,000$-document rotated
e5-small sample. Neighbour overlap excludes the query document itself.}
\label{tab:pq-leakage}
\small
\begin{tabular}{@{}rrrrr@{}}
\toprule
$M$ & Bits & Bytes/vector & Cosine to $E_{\mathrm{rot}}$ & NN overlap@10 \\ \midrule
$24$ & $8$ & $24$ & $0.837$ & $0.384$ \\
$48$ & $8$ & $48$ & $0.953$ & $0.674$ \\
$48$ & $6$ & $36$ & $0.904$ & $0.529$ \\
\bottomrule
\end{tabular}
\end{table}

\begin{figure}[tbp]
\centering
\includegraphics[width=0.58\linewidth]{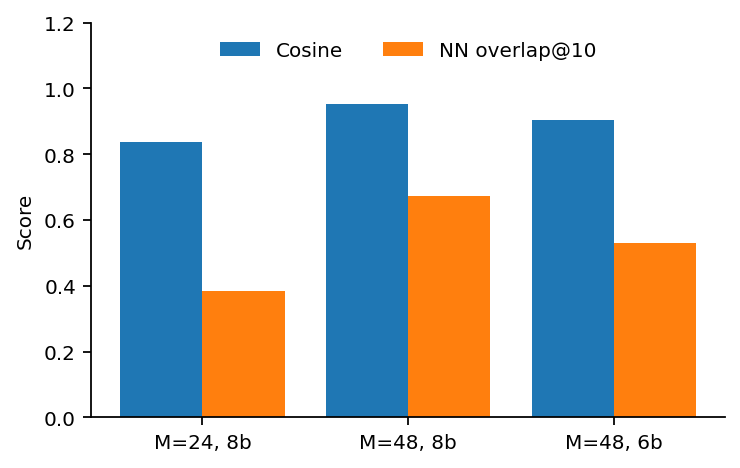}
\caption{Public PQ codes preserve both vector direction and substantial
nearest-neighbour structure in the rotated space.}
\label{fig:pq-leakage}
\end{figure}

These results do not invalidate the query-privacy contribution of CKKS,
but they do narrow the document-privacy interpretation: public PQ
artefacts and known plaintext anchors should be treated as exposed
channels. A system requiring stronger document confidentiality needs a
different static-side primitive or a composition that hides the PQ
artefact itself.

\paragraph{Reference-corpus lookup after alignment.}\label{ssec:reference-attack}

The Procrustes experiment above is an embedding-space test. We next ask
whether the same failure mode can become a text-level leakage channel
without training a generative inverter. The attacker is given the
protected target vectors, $m$ known plaintext/protected pairs for
estimating the rotation, and a reference corpus of candidate paragraphs
with native e5-small embeddings. This models a realistic overlap case:
some private documents may also occur in a public or previously leaked
corpus. We use $500$ target paragraphs, $100\,000$ reference decoys,
five rotation seeds and the same $k=192$ SVD subspace. The overlapping
reference corpus contains the targets plus the decoys; the disjoint
reference corpus removes the targets and is used only as a lexical
nearest-neighbour proxy.

\begin{table}[tbp]
\centering
\caption{Reference-corpus lookup attack after known-plaintext alignment.
The overlapping reference has $500$ targets plus $100\,000$ decoys.
Jaccard is token-set overlap between the target paragraph and the top-1
retrieved paragraph. Values are means over five rotation seeds.}
\label{tab:reference-attack}
\small
\begin{tabular}{@{}rrrrr@{}}
\toprule
Known pairs $m$ & Exact R@1 & Exact R@10 & Jaccard overlap-ref & Jaccard disjoint-ref \\ \midrule
$0$   & $0.000$ & $0.000$ & $0.006$ & $0.006$ \\
$10$  & $0.000$ & $0.009$ & $0.009$ & $0.009$ \\
$25$  & $0.044$ & $0.132$ & $0.055$ & $0.012$ \\
$50$  & $0.535$ & $0.784$ & $0.543$ & $0.018$ \\
$100$ & $0.998$ & $1.000$ & $0.999$ & $0.034$ \\
$192$ & $1.000$ & $1.000$ & $1.000$ & $0.043$ \\
\bottomrule
\end{tabular}
\end{table}

\begin{figure}[tbp]
\centering
\includegraphics[width=0.86\linewidth]{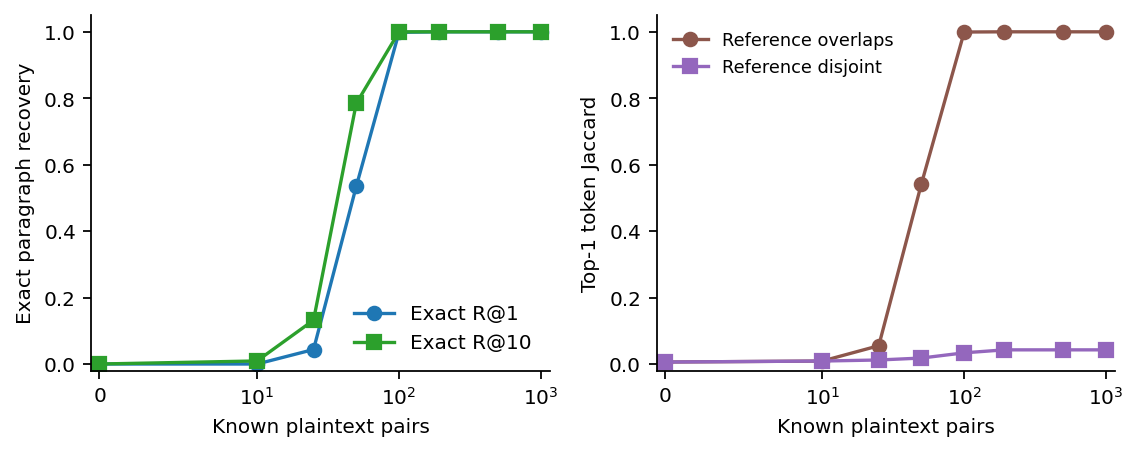}
\caption{Reference-corpus lookup after known-plaintext alignment. When the
reference corpus overlaps the protected collection, alignment turns the
protected vector into an exact paragraph lookup. With a disjoint reference
corpus, this simple token-overlap proxy remains low; semantic leakage
there requires a stronger judge or generative attack.}
\label{fig:reference-attack}
\end{figure}

The overlap case is severe. With only $50$ known pairs, the attacker
recovers the exact source paragraph at $53.5\%$ top-1 and $78.4\%$
top-10 recall among $100\,500$ candidates; with $100$ known pairs,
top-1 recall rises to $99.8\%$. This is the text-level analogue of the
embedding-space Procrustes result and makes the static-side limitation
concrete: if a public or leaked reference corpus overlaps the private
collection, rotation secrecy is not enough. The disjoint-reference
numbers are intentionally reported as a boundary condition: token
Jaccard remains low ($0.034$ at $m=100$ and $0.043$ at $m=192$), so a
non-overlapping reference corpus requires semantic similarity metrics,
human judgement, or a generative inverter before one can claim text
reconstruction.

\paragraph{Aligned Vec2Text inversion stress test.} As a portable
generative-inversion control on the baseline, we also ran an aligned
off-the-shelf \textsc{Vec2Text} corrector (GTR-base, $k=d/2$, up to
$m=500$ Procrustes pairs) against the rotated store. No exact document or
typed PII is recovered in any case, but the \emph{raw} row is also at the
reconstruction floor, so this off-the-shelf corrector certifies nothing;
the open risk is a learned, corpus-adapted decoder. The full setup and
table are reproduced from~\cite{ref_alpha} in
Appendix~\ref{ssec:adaptive-inversion}; the result does not contradict
the exact-lookup failure mode of Section~\ref{ssec:reference-attack}.

\subsection{Baseline utility (recap)}\label{ssec:projection-baselines}\label{ssec:integral}\label{ssec:significance}\label{ssec:beir}

The utility of the global-linear baseline is established in our prior
work~\cite{ref_alpha} and is not reproduced here. In brief, at
$k=d/2$ on the $10^6$-document corpus the PQ$+$CKKS wrapper is
metric-preserving in $\mathrm{span}(V_k)$ (CKKS-decrypted versus exact
scores correlate at Pearson $>0.9999$); the SVD step is a mild,
encoder-dependent denoiser whose apparent Acc@1 gains are \emph{not}
statistically significant and which \emph{significantly reduces} nDCG@10 on
graded-relevance BEIR, so $k$ is a tunable accuracy cost rather than a free
denoiser, met within the end-to-end accuracy budget only for
retrieval-trained $d\geq768$ encoders. \textsc{Shard} reuses the baseline's
ct--pt arithmetic but not its half-space score: in exact plaintext algebra it
reranks in the full space (Section~\ref{sec:shard}). Its own retrieval quality
is evaluated in Section~\ref{ssec:shard-utility}.

\subsection{Baseline latency (recap)}\label{ssec:latency}

End-to-end latency of the global-linear baseline is reported in full
in~\cite{ref_alpha}; the point that matters for \textsc{Shard} is
architectural. The pipeline reranks only the $K_{\mathrm{cands}}{=}40$
stage-1 candidates, not $10^6$, giving a server-side $p_{95}$ of
$219$--$356$~ms and $\approx578$~ms end-to-end once HTTP/JSON framing is
added; the local PQ stage---not CKKS---reduces $10^6$ documents to $40$, so
the cryptographic cost is independent of corpus size $N_{\mathrm{docs}}$.
\textsc{Shard} adds a two-stage rescoring whose marginal online cost is
measured in Section~\ref{ssec:shard-cost}.

\subsection{SHARD: utility and online cost}\label{ssec:shard-utility}

The experiments above characterise the global-linear baseline. We now
evaluate \textsc{Shard} on the same cached embeddings, computing the
geometry with \texttt{numpy}. This verifies the plaintext identity; it is not
a substitute for measuring CKKS approximation in an end-to-end prototype.

The corrected protocol exposes an error in the earlier comparison. Documents
use $V^\top(x-\mu)$, but a ranking-preserving scoring query must use
$V^\top q$. Centring both sides adds a document-dependent term and changes
the ranking. Experiment~23 evaluates raw, legacy-centred, corrected
half-PCA, and corrected two-stage scores on every judged query in six
BEIR~\cite{ref_beir}
and four MIRACL cells, with $10^4$ paired-bootstrap resamples. In an
independent float64 check the corrected full score and raw score have the
same complete ordering in all tested queries (maximum numerical error
$7.71\times10^{-16}$). Table~\ref{tab:shard-beir} shows that the $d/4$ SHARD
shortlist followed by corrected full reranking remains within $0.009$ nDCG
of raw; its small residual delta is shortlist recall, not score distortion.

\paragraph{Why this is not merely ``do not truncate''.} A reader might
object that the baseline could also rerank full-dimensionally by storing
the untruncated $E_{\mathrm{rot}}$. It cannot do so for free: an untruncated
plaintext store has a \emph{single global key that cancels for all pairs},
so the server recovers the entire $d$-dimensional neighbour graph. The
point of \textsc{Shard} is narrower: full-dimensional reranking need not
create a \emph{globally} comparable residual geometry. It does create exact
within-cell geometry and stable norm signatures, so it is not privacy-neutral.
The utility result and the leakage result must be read together.

\begin{table}[tbp]
\centering
\caption{\textsc{Shard} utility on three BEIR datasets (SciFact, NFCorpus,
ArguAna; nDCG@10, $K_{\mathrm{cands}}{=}200$). Each protected column reports
the change $\Delta$ vs.\ raw with a paired-bootstrap $95\%$ CI over
per-query nDCG@10 ($10^4$ resamples, seed $2026$). ``Centred full'' isolates
the legacy query-centring error; ``corrected half'' isolates genuine
half-PCA truncation; \textsc{Shard} uses corrected full scoring after a
centred $d/4$ router.}
\label{tab:shard-beir}
\scriptsize
\setlength{\tabcolsep}{4pt}
\begin{tabular}{@{}llrrrr@{}}
\toprule
Encoder & Dataset & raw & centred full $\Delta$\,[95\% CI] & corrected half $\Delta$\,[95\% CI] & \textsc{Shard} $d/4$ $\Delta$\,[95\% CI]\\
\midrule
e5-base  & SciFact  & $0.637$ & $-0.051\,[{-}.072,{-}.031]$ & $-0.015\,[{-}.028,{-}.002]$ & $\mathbf{+0.001\,[.000,.003]}$\\
e5-base  & NFCorpus & $0.327$ & $-0.022\,[{-}.031,{-}.012]$ & $0.000\,[{-}.005,.005]$ & $\mathbf{-0.001\,[{-}.002,.000]}$\\
e5-base  & ArguAna  & $0.346$ & $-0.080\,[{-}.089,{-}.071]$ & $-0.008\,[{-}.012,{-}.003]$ & $\mathbf{-0.001\,[{-}.002,.000]}$\\
e5-small & SciFact  & $0.598$ & $-0.039\,[{-}.059,{-}.020]$ & $-0.049\,[{-}.068,{-}.030]$ & $\mathbf{-0.008\,[{-}.020,.001]}$\\
e5-small & NFCorpus & $0.302$ & $-0.007\,[{-}.016,.002]$ & $-0.012\,[{-}.019,{-}.005]$ & $\mathbf{-0.002\,[{-}.004,.000]}$\\
e5-small & ArguAna  & $0.358$ & $-0.038\,[{-}.046,{-}.030]$ & $-0.030\,[{-}.038,{-}.022]$ & $\mathbf{-0.002\,[{-}.004,.000]}$\\
\bottomrule
\end{tabular}
\end{table}

\begin{figure}[tbp]
\centering
\includegraphics[width=0.90\linewidth]{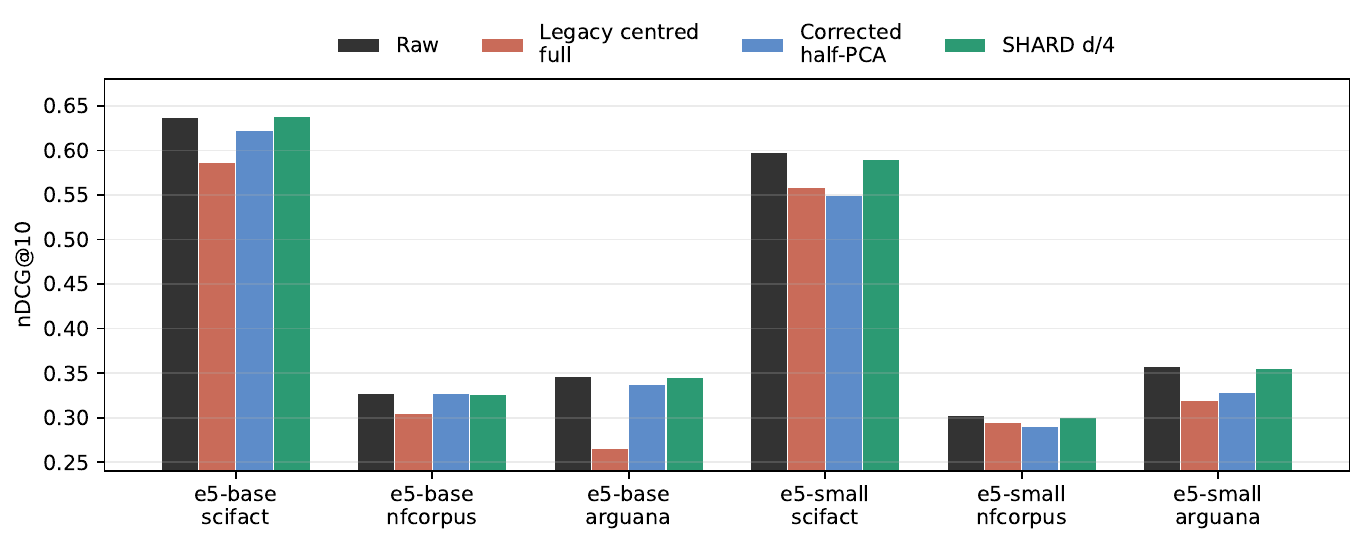}
\caption{BEIR utility after correcting query centring. The corrected
half-PCA gap is smaller than the legacy result but remains dataset-dependent;
the $d/4$ router followed by corrected full scoring stays close to raw.}
\label{fig:shard-beir}
\end{figure}

\noindent
The corrected half-PCA result changes the interpretation. Truncation remains
significantly harmful on five cells, but e5-base/NFCorpus is effectively
unchanged ($\Delta=-0.00002$); much of the previous baseline gap came from
query centring rather than discarded dimensions. The SHARD $d/4$ deltas stay
between $-0.008$ and $+0.001$, confirming the rank-preserving score and
showing the remaining shortlist effect without overstating truncation.

\paragraph{Multilingual replication at $10^5$ scale (MIRACL).}
To test whether the BEIR pattern holds on a large, \emph{multilingual}
benchmark scored with a graded ranking metric, we repeat the full-corpus
protocol on MIRACL---the canonical multilingual IR benchmark (binary
judgements scored with nDCG@10)---for two typologically distinct languages
with manageable corpora: Swahili (Bantu, Latin script; $131{,}924$ passages,
$482$ dev queries) and Bengali (Indo-Aryan, Bengali script; $297{,}265$
passages, $411$ queries). Table~\ref{tab:shard-miracl} shows the same
pattern at $10^5$-document scale. Corrected half-PCA loses $0.012$--$0.016$
nDCG on e5-base and $0.039$--$0.053$ on e5-small; the legacy centring error
is separately visible in every row. The corrected $d/4$ SHARD pipeline is
equal to raw on e5-base and within $0.006$--$0.009$ on e5-small. This
reinforces $d/4$ as a reasonable router, while attributing each loss to the
right source.

\begin{table}[tbp]
\centering
\caption{\textsc{Shard} utility on MIRACL (multilingual, nDCG@10,
$K_{\mathrm{cands}}{=}200$): full-corpus retrieval over the dev queries of two
typologically distinct languages. Columns as in Table~\ref{tab:shard-beir}:
$\Delta$ vs.\ raw with a paired-bootstrap $95\%$ CI over per-query nDCG@10
($10^4$ resamples, seed $2026$). $|C|$ is the corpus size; MIRACL uses binary
judgements scored with the graded nDCG@10 metric.}
\label{tab:shard-miracl}
\scriptsize
\setlength{\tabcolsep}{4pt}
\begin{tabular}{@{}llrrrr@{}}
\toprule
Encoder & Language ($|C|$) & raw & centred full $\Delta$\,[95\% CI] & corrected half $\Delta$\,[95\% CI] & \textsc{Shard} $d/4$ $\Delta$\,[95\% CI]\\
\midrule
e5-base  & Swahili ($132$k) & $0.707$ & $-0.019\,[{-}.032,{-}.007]$ & $-0.012\,[{-}.020,{-}.003]$ & $\mathbf{+0.000\,[.000,.001]}$\\
e5-base  & Bengali ($297$k) & $0.727$ & $-0.031\,[{-}.043,{-}.020]$ & $-0.016\,[{-}.025,{-}.007]$ & $\mathbf{0.000\,[.000,.000]}$\\
e5-small & Swahili ($132$k) & $0.680$ & $-0.036\,[{-}.051,{-}.021]$ & $-0.053\,[{-}.068,{-}.037]$ & $\mathbf{-0.006\,[{-}.015,.001]}$\\
e5-small & Bengali ($297$k) & $0.688$ & $-0.033\,[{-}.046,{-}.020]$ & $-0.039\,[{-}.055,{-}.024]$ & $\mathbf{-0.009\,[{-}.018,{-}.001]}$\\
\bottomrule
\end{tabular}
\end{table}

\paragraph{Online cost and the $C$ trade-off.}\label{ssec:shard-cost}

\textsc{Shard} is not free relative to the single-query baseline, and we
first expose the structural active-cell cost before measuring the concrete
CKKS path. The client sends one encrypted
residual query \emph{per active cell}---a cell touched by the stage-1
short-list---so the per-query upload and client-side encryption scale with
the number of distinct cells the short-list spans. Table~\ref{tab:shard-cost}
reports the active-cell counts. At $C{=}256$, $K_{\mathrm{cands}}{=}200$ the
short-list spans about $30$ cells. The historical $6$~MB figure obtained by
multiplying this count by the baseline ciphertext size is deliberately kept
as a first-order model; the measured serialization below supersedes it for
the implemented layouts. The structural trade-off remains: larger $C$ slows the accumulation
of in-cell alignment evidence but costs more active-cell query transforms. The
numbers are consistent across encoders (e5-base spans $\approx 9$--$35$
cells over the same grid). Active cells grow because $k$-means cells are
finer than a query's neighbourhood; a deployment tunes $(C,
K_{\mathrm{cands}}, d_{\mathrm{pub}})$ to balance alignment resistance,
stage-1 recall and upload, and the micro-key extreme ($C{=}N$) is the
endpoint where every candidate is its own cell ($\approx K_{\mathrm{cands}}$
queries).

\begin{table}[tbp]
\centering
\caption{Structural active-cell cost (e5-small, $d_{\mathrm{pub}}=d/4$).
The last two columns retain the earlier one-$0.21$~MB-ciphertext model for
comparison; Table~\ref{tab:shard-ckks} reports actual serialized traffic.}
\label{tab:shard-cost}
\small
\begin{tabular}{@{}lrrrr@{}}
\toprule
$K_{\mathrm{cands}}$ & \multicolumn{2}{c}{active cells (mean / p95)} & upload (MB) & $\times$ baseline\\
\cmidrule(lr){2-3}
 & $C{=}64$ & $C{=}256$ & ($C{=}256$) & ($C{=}256$)\\
\midrule
$40$  & $7.2 / 13$  & $11.6 / 19$ & $2.4$ & $11.6\times$\\
$100$ & $11.3 / 19$ & $20.0 / 32$ & $4.2$ & $20.0\times$\\
$200$ & $15.7 / 26$ & $30.1 / 48$ & $6.3$ & $30.1\times$\\
\bottomrule
\end{tabular}
\end{table}

\paragraph{Measured CKKS and block-SIMD packing.}\label{ssec:shard-ckks}
Experiment~26 implements the residual reranker with real CKKS operations in
TenSEAL/SEAL~\cite{ref_tenseal}, using cached multilingual-e5 SciFact
document/query embeddings for both encoder sizes. It is a systems benchmark
on this real embedding geometry, not a claim of ten-dataset CKKS coverage.
The public server context contains the public and Galois keys but no secret
key. We use $N_{\mathrm{poly}}=8192$, a
$[60,40,40,60]$ coefficient-modulus chain and scale $2^{40}$; the provisioned
server context is $34.08$~MB and is excluded from per-query traffic. A width
$B$ layout concatenates $B$ active-cell query blocks in one ciphertext. For a
candidate in block $b$, the plaintext operand is zero outside that block, so
one real ciphertext--plaintext dot product returns its residual contribution.
The implementation adds the public-prefix score as a plaintext scalar to that
encrypted response. The client then decrypts one encrypted full score per
candidate.

The benchmark contains 315 measured trials over two encoders,
$K_{\mathrm{cands}}\in\{16,64,128\}$, three seeds, five repetitions and every
layout that fits the $4096$ CKKS slots. Each timing includes client transform,
packing, encryption and serialization; server deserialization, plaintext-mask
construction, ct--pt evaluation and response serialization; and client
deserialization and decryption. Table~\ref{tab:shard-ckks} gives the most
demanding $K=128$ cells. Across the whole grid the maximum score error is
$2.29\times10^{-6}$, there are no top-1 flips, top-10 overlap never falls below
$1.000$, and the minimum Kendall $\tau$ is $0.999754$. The largest p95
point-sampled RSS increment is $46.8$~MB (e5-base, $K=128$, width~1); this is
an approximate sampled increment rather than a high-frequency peak trace.

\begin{table}[tbp]
\centering
\caption{Actual CKKS reranking at $K_{\mathrm{cands}}=128$ (15 trials per
row). Traffic is the median serialized decimal MB; latency is the in-process
online-path p50/p95. ``ct'' is the median number of uploaded query
ciphertexts.}
\label{tab:shard-ckks}
\small
\begin{tabular}{@{}llrrrrr@{}}
\toprule
Encoder & pack $B$ & query ct & upload & response & p50 (s) & p95 (s) \\ \midrule
e5-small & $1$ & $30$ & $10.03$ & $30.13$ & $2.415$ & $2.504$ \\
e5-small & $8$ & $4$  & $1.34$  & $30.13$ & $3.049$ & $3.180$ \\
e5-base  & $1$ & $31$ & $10.36$ & $30.13$ & $2.655$ & $2.741$ \\
e5-base  & $4$ & $8$  & $2.67$  & $30.13$ & $3.040$ & $3.136$ \\
\bottomrule
\end{tabular}
\end{table}

\begin{figure}[tbp]
\centering
\includegraphics[width=0.98\linewidth]{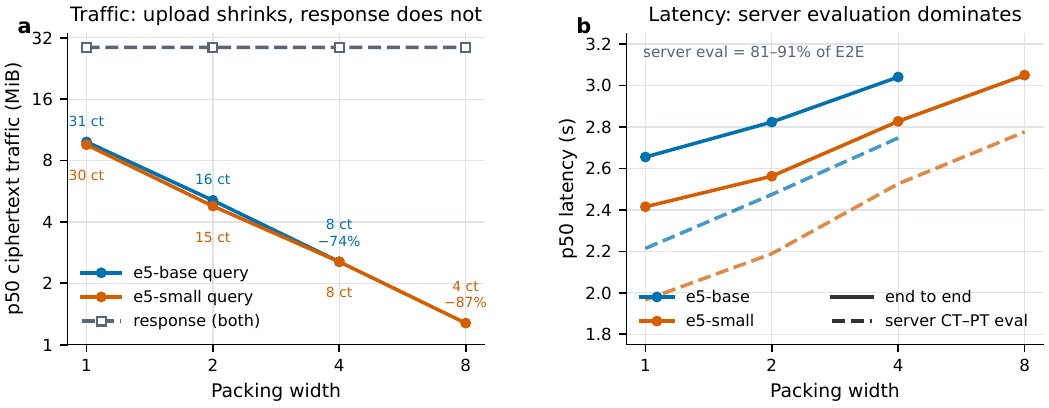}
\caption{Measured block-SIMD trade-off at $K=128$. Packing cuts query upload
by $74.2\%$ (e5-base, width~4) and $86.7\%$ (e5-small, width~8), but the
separate encrypted response per candidate remains $30.13$ decimal MB
($28.7$~MiB on the axis). TenSEAL's
dot reduction over the larger logical vector increases p50 latency by
$14.5\%$ and $26.3\%$, respectively.}
\label{fig:shard-ckks}
\end{figure}

This result answers the implementation question more sharply than a byte
model could. Block-SIMD packing is implemented and materially reduces upload
and client encryption, but it is not a latency win in this layout: server
evaluation occupies $81$--$91\%$ of end-to-end time, while one response
ciphertext is still emitted for every candidate. Packed multi-score output,
real network RTT/TLS, concurrent clients and an ANN service are therefore
systems work still to be done; they are not silently folded into these local
numbers.

\paragraph{Choosing $C$ from measured recovery curves.}
The earlier draft derived a supposedly secure cell count by assuming that an
attack fails until a cell contains $d_{\mathrm{priv}}$ anchors. Experiment~24
rejects that premise: partial estimators exploit a lower-rank anchor span.
There is therefore no bandwidth-optimal secure $C$ in this model. A
deployment can only choose an empirical operating point by comparing the
active-cell cost in Table~\ref{tab:shard-cost} with an acceptable recovery
curve under a specified disclosure budget. Increasing $C$ dilutes a diffuse
anchor stream; it does not stop an attacker who can place anchors in one
target cell.

\subsection{SHARD under partial alignment}\label{ssec:shard-align}

We repeat the residual-only known-pair experiment without the artificial
full-rank gate. An observer receives $m$ nested pairs $(r_i,z_i)$, fits a
reverse map in every target cell containing at least one pair, and matches
$256$ held-out targets against a $10\,000$-document native residual gallery.
The e5-small pool has $100\,000$ documents, $d_{\mathrm{pub}}=96$ and
$d_{\mathrm{priv}}=288$; we evaluate $C\in\{1,16,64,256\}$ over three
key/disclosure seeds. Confidence intervals use a hierarchical
seed-then-target bootstrap with $500$ repetitions.

The methods are rank-deficient orthogonal Procrustes, a Moore--Penrose
minimum-norm linear map, ridge with both zero and cross-validated
regularisation, and the polar projection of the linear map. All methods run
for every positive in-cell anchor count. This is a stronger and fairer game
than declaring an untouched keyed vector whenever the anchor matrix is not
full rank.

\begin{proposition}[Identifiable anchor span]\label{prop:anchor}
Let $z_j=H r_j$ for an unknown $H\in O(d_{\mathrm{priv}})$ and let
$\mathcal S=\operatorname{span}\{r_1,\ldots,r_m\}$ have rank $s$. The pairs
determine the action of $H$ on $\mathcal S$, but leave an arbitrary
orthogonal extension on $\mathcal S^\perp$. Thus exact identification of a
generic full key requires $s=d_{\mathrm{priv}}$, whereas the minimum-norm
reverse map applied to a new target recovers its projection
$P_{\mathcal S}r$. For an isotropic target,
\[
 \frac{\mathbb E\lVert P_{\mathcal S}r\rVert_2^2}
      {\mathbb E\lVert r\rVert_2^2}=\frac{s}{d_{\mathrm{priv}}}.
\]
\end{proposition}
\begin{proof}
The paired vectors define an isometry between $\mathcal S$ and
$H\mathcal S$. Any two orthogonal maps that agree on $\mathcal S$ and differ
only on $\mathcal S^\perp$ produce the same observed pairs. Writing the row
matrices as $Z=R H^\top$, the minimum-norm least-squares reverse map is
$Z^+R=H P_{\mathcal S}$, hence a new row $z^\top=r^\top H^\top$ maps to
$r^\top P_{\mathcal S}$. The isotropic expectation is the trace identity
$\mathbb E[r^\top P_{\mathcal S}r]/\mathbb E[r^\top r]=s/d$.
\end{proof}

The proposition separates two questions that the previous analysis
conflated. A complete key remains underdetermined below full rank, but record
matching may need only a discriminative projection.
Table~\ref{tab:shard-align} makes the distinction stark. Minimum-norm OLS
exceeds $0.9$ R@1 with only about $32$--$36$ anchors in the average target
cell, while full-rank coverage is zero in every row. The required
\emph{global} budget still grows approximately with $C$ under diffuse
disclosure ($32,512,2048,8192$), so compartmentalisation is real; its
absolute security interpretation is much weaker than the former
$d_{\mathrm{priv}}$-anchors-per-cell claim.

\begin{table}[tbp]
\centering
\caption{Strongest partial-alignment result (minimum-norm OLS, e5-small).
The reported global $m$ is the first tested point with mean R@1 at least
$0.9$. CIs are hierarchical $95\%$ intervals. Exact full-key rank is absent
at every operating point even though record matching is already reliable.}
\label{tab:shard-align}
\small
\begin{tabular}{@{}rrrrrr@{}}
\toprule
$C$ & global $m$ & mean anchors/cell & R@1 [95\% CI] &
recovered cosine & full-rank coverage \\ \midrule
$1$   & $32$     & $32.0$ & $0.988\,[0.979,0.996]$ & $0.356$ & $0.000$ \\
$16$  & $512$    & $35.8$ & $0.939\,[0.921,0.957]$ & $0.394$ & $0.000$ \\
$64$  & $2\,048$ & $34.8$ & $0.953\,[0.932,0.970]$ & $0.412$ & $0.000$ \\
$256$ & $8\,192$ & $33.7$ & $0.961\,[0.944,0.976]$ & $0.443$ & $0.000$ \\
\bottomrule
\end{tabular}
\end{table}

\begin{figure}[tbp]
\centering
\includegraphics[width=\linewidth]{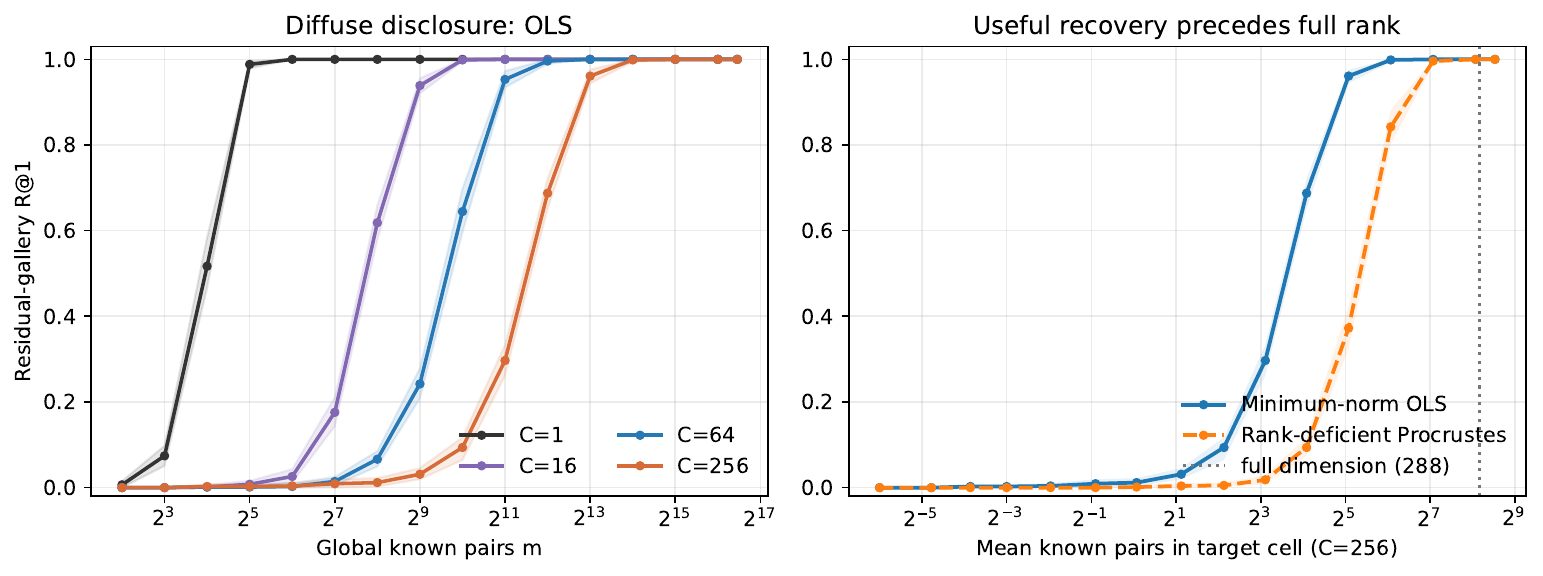}
\caption{Partial residual alignment without a full-rank gate. Minimum-norm
OLS extracts the observed anchor span and reaches high gallery R@1 far before
the complete cell key is identifiable. Increasing $C$ shifts the global
disclosure curve by spreading anchors across cells; it does not create a hard
threshold.}
\label{fig:shard-align}
\end{figure}

\paragraph{Attacker comparison and targeted disclosure.}\label{ssec:shard-learned}
Ridge with $\alpha=0$ is numerically identical to the Moore--Penrose result;
cross-validated ridge differs by less than $2\times10^{-6}$. Orthogonal
Procrustes and the polar projection are weaker below full rank but still
reach R@1 above $0.9$ at roughly $128$ in-cell anchors: global
$m=128,2\,048,8\,192,32\,768$ for $C=1,16,64,256$, respectively, again with
zero full-rank coverage at the threshold. The earlier fixed
$\operatorname{Ridge}(\alpha=1)$ comparison therefore did not establish that
Procrustes was the strongest observer, and that claim is withdrawn.
Likewise, a covariance-matching control is not a full implementation of an
unsupervised translation attack.

For a targeted observer, only the in-cell count matters: the OLS rows above
indicate that approximately $32$--$36$ well-spread pairs in the victim cell
already give high residual-only R@1 in this gallery, independently of the
global cell count. The experiment excludes the unchanged public prefix, norm
and cross-release Gram channels; Section~\ref{ssec:shard-micro} shows that the
full server view is easier to link.

\paragraph{Checkpoint-compatible text outcome under a strengthened observer.}
\label{ssec:shard-text}
Embedding-space R@1 does not by itself say whether the recovered vector yields
readable text. Experiment~29 therefore sends actual SHARD-derived views to the
official public GTR-base \textsc{Vec2Text} inversion/correction checkpoints
\cite{ref_gtr,ref_vec2text}, using their mask-aware mean-pooling path rather
than the incorrect CLS extraction used by an earlier diagnostic. The corpus
contains $2{,}400$ cached AG News records and $600$ controlled synthetic
records, truncated to the checkpoint's 32-token input. We fit a full
$768$-dimensional PCA and evaluate $(d_{\mathrm{pub}},C)=(192,16)$ and
$(384,32)$, two independent cell keys per geometry, 12 targets from the
globally largest public cell, eight correction steps, and nested
$m\in\{8,16,32,64\}$ in-cell pairs. The views are raw, prefix-only, a keyed
residual treated as if its key were the identity, exact-key oracle,
minimum-norm OLS, and rank-deficient Procrustes.

This is deliberately a strengthened, checkpoint-compatible stress test. The
observer is granted the exact PCA basis $V$ and mean $\mu$, leaving only the
cell key $H_c$ unknown or partially learned; it is therefore stronger than the
honest server in our primary threat model and isolates the outcome of residual
key recovery. The official corrector expects native unnormalised mean-pooled
GTR vectors (mean norm $1.047$ here), so this case is not the literal
unit-normalised e5 retrieval instance. Raw and exact-key oracle are identical
geometry controls, not upper bounds on a non-monotone neural decoder. The
intervals are conditional, descriptive bootstraps over two fixed geometry
designs and their targets, not population-level significance claims.

Table~\ref{tab:shard-text} gives the globally largest-cell outcome. Averaged
equally over the two geometry designs, the unknown-key view lowers token-F1
from $0.665$ for the exact-geometry control to $0.242$. This is the positive
no-anchor effect that survives the stronger raw control. It is not isolation:
the public prefix alone retains F1 $0.330$ at $d_{\mathrm{pub}}=192$ and
$0.536$ at $d_{\mathrm{pub}}=384$. With only eight in-cell pairs,
minimum-norm OLS reaches $0.343$/$0.557$, essentially returning to the
prefix-only outcome. Increasing the pair budget improves geometric cosine but
does not make neural decoding monotone; at $m=64$, F1 is
$0.373$/$0.516$. Rank-deficient Procrustes is weaker in this grid.

\begin{table}[H]
\centering
\caption{Checkpoint-compatible GTR text outcome in the globally largest
public cell (12 AG News targets per geometry, two keys). Each cell reports
input cosine / token-F1 / BLEU. Values are descriptive means; raw and
exact-key oracle are numerically identical, so they share one row.}
\label{tab:shard-text}
\small
\begin{tabular}{@{}llcc@{}}
\toprule
Observer view & $m$ & $d_{\mathrm{pub}}{=}192$, $C{=}16$ & $d_{\mathrm{pub}}{=}384$, $C{=}32$ \\ \midrule
raw / exact key & --- & $1.000/.670/.207$ & $1.000/.661/.214$ \\
prefix only & --- & $.888/.330/.042$ & $.985/.536/.094$ \\
unknown key & --- & $.789/.152/.020$ & $.971/.331/.047$ \\
minimum-norm OLS & $8$ & $.892/.343/.050$ & $.986/.557/.114$ \\
minimum-norm OLS & $64$ & $.916/.373/.049$ & $.990/.516/.107$ \\
Procrustes & $64$ & $.839/.192/.022$ & $.979/.410/.062$ \\
\bottomrule
\end{tabular}
\end{table}

A separately labelled, evaluator-selected cell makes the controlled PII
diagnostic measurable; it is not an attacker-discoverable primary cohort.
Across its duplicated key releases, exact repeated-name recall is $.792$ for
raw, $.458$ for the unknown-key view and $.708$ for OLS at $m=8$. These names
are not unique (307 values among 600 records), and some target names also occur
among anchors. More importantly, exact recovery of unique visible email
addresses is $0/48$ and of phone numbers $0/34$ for \emph{every} view,
including raw and exact-key controls; card numbers have no visible denominator
after truncation. The PII diagnostic therefore lacks a positive unique-ID
control and supports no PII-protection claim. It is reported to prevent a
favourable token-F1 result from being relabelled as sensitive-attribute
security.

\FloatBarrier
\subsection{SHARD: the public channel leaks coarse structure}\label{ssec:shard-leak}

The baseline's public PQ index sits on the full protected space and
preserves $67\%$ of exact top-10 neighbours (Section~\ref{ssec:alignment-pq}).
\textsc{Shard}'s public channel is only the short prefix $u$. We measure how
much true (full-space) top-10 neighbour structure each public channel
reveals (NN-overlap@10, e5-small, $100\,000$ documents, $1000$ probes;
Fig.~\ref{fig:shard-leak}). A $d/4$ prefix reveals NN-overlap $0.55$ and a
$d/8$ prefix only $0.36$ (and $0.20$ at $d/16$), against $0.76$ for the
baseline's exposed SVD $k{=}d/2$ geometry. The public stage-1 channel thus
discloses coarse/topic neighbours but not fine identity, which lives in the
keyed residual. The pattern reproduces on e5-base: a $d/8$ prefix leaks
NN-overlap $0.58$ and a $d/4$ prefix $0.77$, against $0.94$ for the
baseline geometry; the within-cell fraction is $0.55$/$0.46$ at
$C{=}64$/$256$.

\begin{figure}[tbp]
\centering
\includegraphics[width=0.66\linewidth]{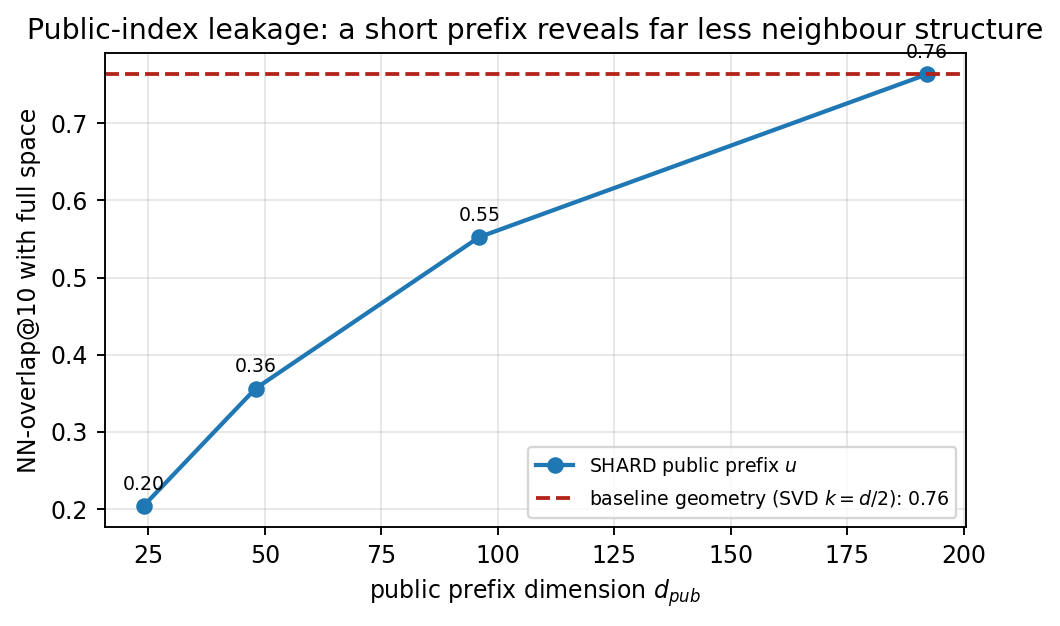}
\caption{Public-index leakage (e5-small). The \textsc{Shard} public prefix
$u$ reveals far less of the true top-10 neighbour structure than the
baseline's exposed SVD $k{=}d/2$ geometry; a $d/8$ prefix discloses about
half as much.}
\label{fig:shard-leak}
\end{figure}

We are explicit about the residual leak. Within a cell the keys cancel, so
the server can compute exact same-cell similarities; the fraction of a
document's true top-10 neighbours that fall in its own cell---hence remain
recoverable---is $0.50$ at $C{=}64$ and $0.41$ at $C{=}256$, against $1.00$
for the baseline's single global key (which cancels for \emph{all} pairs).
More cells localise the leak further, and the micro-key variant removes it
entirely. \textsc{Shard} therefore shrinks, but does not zero, the
neighbour graph available to an honest-but-curious server.

\paragraph{Baseline: a PQ-only hardened store.} A natural alternative to
cell-keying is to publish product-quantization codes of the rotated store
and rerank from them---a \emph{PQ-only} defense. We compare it on the same
data and leakage metric (Table~\ref{tab:shard-pq-only}); the result is an
honest trade-off, not a clean win. At small code sizes PQ-only leaks
\emph{less} coarse neighbour structure than \textsc{Shard}'s $d/4$ prefix
(NN-overlap@10 $0.27$--$0.63$ vs.\ $0.55$--$0.77$), because lossy
quantization blurs the high-variance directions the prefix exposes. But that
loss is the catch: PQ-only reranks from a reconstruction of cosine only
$0.67$--$0.83$ to the true vector, so retrieval quality degrades, whereas
\textsc{Shard}'s per-cell orthogonal keys cancel and rerank \emph{exactly}
(cosine $1.000$)---the raw-nDCG recovery of Section~\ref{ssec:shard-utility}
that PQ-only cannot match. The two are different points on the
privacy/utility frontier: PQ-only buys a smaller footprint and lower coarse
leak at a retrieval-quality cost; \textsc{Shard} preserves the plaintext score
and instead localises the residual leak to cells (and removes it at the
micro-key limit, Section~\ref{ssec:shard-micro}).

\begin{table}[tbp]
\centering
\caption{A PQ-only hardened store vs.\ \textsc{Shard} ($100\,000$ documents,
$1000$ probes). Rerank fidelity is the cosine of the channel used for
scoring to the true vector (\textsc{Shard}'s cell-keyed residual cancels
exactly; PQ-only is lossy); NN-overlap@10 is the full-space top-10 structure
recoverable from the public channel.}
\label{tab:shard-pq-only}
\small
\begin{tabular}{@{}llrr@{}}
\toprule
Encoder & Public channel & rerank fidelity & NN-overlap@10 \\ \midrule
e5-small & \textsc{Shard} $d/4$ prefix & plaintext ($1.000$) & $0.55$ \\
e5-small & PQ-only $M{=}24$ ($24$\,B) & $0.67$ & $0.27$ \\
e5-small & PQ-only $M{=}48$ ($48$\,B) & $0.83$ & $0.49$ \\ \midrule
e5-base & \textsc{Shard} $d/4$ prefix & plaintext ($1.000$) & $0.77$ \\
e5-base & PQ-only $M{=}48$ ($48$\,B) & $0.67$ & $0.41$ \\
e5-base & PQ-only $M{=}96$ ($96$\,B) & $0.83$ & $0.63$ \\
\bottomrule
\end{tabular}
\end{table}

\paragraph{Micro-keys remove shared residual edges, not cross-release linkage.}\label{ssec:shard-micro}

Per-document keys do eliminate shared-key residual comparisons: the fraction
of a document's true residual top-10 neighbours directly recoverable from
same-key products falls from $0.25$ at $C{=}64$ and $0.18$ at $C{=}256$ to
$0.00$. That within-release result must not be confused with unlinkability.

Experiment~25 creates two independently re-keyed snapshots of the same
e5-small records, independently permutes row handles, removes persistent
document identifiers, and evaluates both a residual-only view and the full
server view. At $N=10{,}000$ the earlier raw cross-key cosine is indeed at
chance, but stronger invariant features reverse the conclusion
(Table~\ref{tab:shard-micro}). For $C{=}64$, monotone one-to-one assignment
by residual norm links $99.6\%$ of records and a permutation-invariant
within-cell Gram signature links $99.9\%$. For micro-keys the Gram channel
vanishes, yet norm assignment still links $99.4\%$. The unchanged prefix
links $99.9\%$ in both cases. Cell-level Gram matching remains above $99.9\%$
under fp16, int8 and a diagnostic $1\%$ Gaussian perturbation. These are
same-population snapshot results without insert/delete churn, not a universal
linkage benchmark, but they decisively reject the former unlinkability claim.

\begin{table}[tbp]
\centering
\caption{Cross-release linkage after independent row permutation
(e5-small, clean snapshots, $N=10{,}000$, mean over three seeds). The weak
cross-key cosine control is near chance, but orthogonal invariants and the
unchanged prefix link almost every record.}
\label{tab:shard-micro}
\small
\begin{tabular}{@{}llrr@{}}
\toprule
Keying & Linkage method & top-1 & AUC \\ \midrule
cell, $C{=}64$ & cell-label chance & $0.0064$ & $0.500$ \\
cell, $C{=}64$ & raw cross-key cosine & $0.0057$ & $0.501$ \\
cell, $C{=}64$ & norm-rank assignment & $0.9957$ & --- \\
cell, $C{=}64$ & Gram signature & $0.9993$ & $1.000$ \\
cell, $C{=}64$ & unchanged prefix & $0.9993$ & $1.000$ \\ \midrule
micro-key & global chance & $0.0001$ & $0.500$ \\
micro-key & raw cross-key cosine & $0.0003$ & $0.497$ \\
micro-key & norm-rank assignment & $0.9943$ & --- \\
micro-key & unchanged prefix & $0.9993$ & $1.000$ \\
\bottomrule
\end{tabular}
\end{table}

Experiment~28 makes this game less artificial by giving each release
$N\in\{2{,}000,10{,}000\}$ rows but retaining only $25\%$, $50\%$, $75\%$,
$90\%$ or $100\%$ of the underlying documents. Release-local row order is
random, identifiers and overlap labels are hidden from the attacker, and the
two releases contain disjoint insertion/deletion sets around the common
subset. We repeat the game for e5-small and e5-base, three independent
release pairs, clean float32, fp16 and release-local symmetric int8, using
features standardised only from the two observed marginals. Thus no true pair
or overlap label is used to build the linker.

Churn weakens Gram signatures because a row's within-cell neighbourhood has
changed, but it does not repair the unchanged channels. At $N=10{,}000$ in
clean releases, public-prefix R@1 for persistent rows is at least $0.9997$ and
residual-norm R@1 at least $0.9996$ at every overlap. Even at $50\%$ overlap,
fp16 norm linkage reaches $0.9632$/$0.9782$ for e5-small/e5-base. Int8 lowers
those values to $0.4505$/$0.5522$, but the corresponding pairwise AUC remains
$0.9930$/$0.9951$. Quantisation therefore makes exact scalar assignment less
stable without erasing the signal.

\begin{table}[tbp]
\centering
\caption{Common-document R@1 under insert/delete churn ($N=10{,}000$, clean,
three-seed mean). Values are the combined norm/Gram/neighbour signature;
parentheses give Gram alone. Prefix and norm remain at R@1$\ge0.9996$ for all
clean overlap conditions.}
\label{tab:shard-churn}
\small
\begin{tabular}{@{}llrrrrr@{}}
\toprule
Encoder & keying & $25\%$ & $50\%$ & $75\%$ & $90\%$ & $100\%$ \\ \midrule
e5-small & cell & $.148\,(.134)$ & $.315\,(.285)$ & $.650\,(.604)$ & $.913\,(.884)$ & $1.000\,(1.000)$ \\
e5-base  & cell & $.140\,(.134)$ & $.300\,(.273)$ & $.638\,(.585)$ & $.913\,(.875)$ & $1.000\,(1.000)$ \\
e5-small & micro & $.016\,(.015)$ & $.020\,(.018)$ & $.020\,(.015)$ & $.017\,(.015)$ & $.021\,(.017)$ \\
e5-base  & micro & $.019\,(.017)$ & $.021\,(.019)$ & $.019\,(.017)$ & $.020\,(.018)$ & $.021\,(.018)$ \\
\bottomrule
\end{tabular}
\end{table}

\begin{figure}[tbp]
\centering
\includegraphics[width=0.96\linewidth]{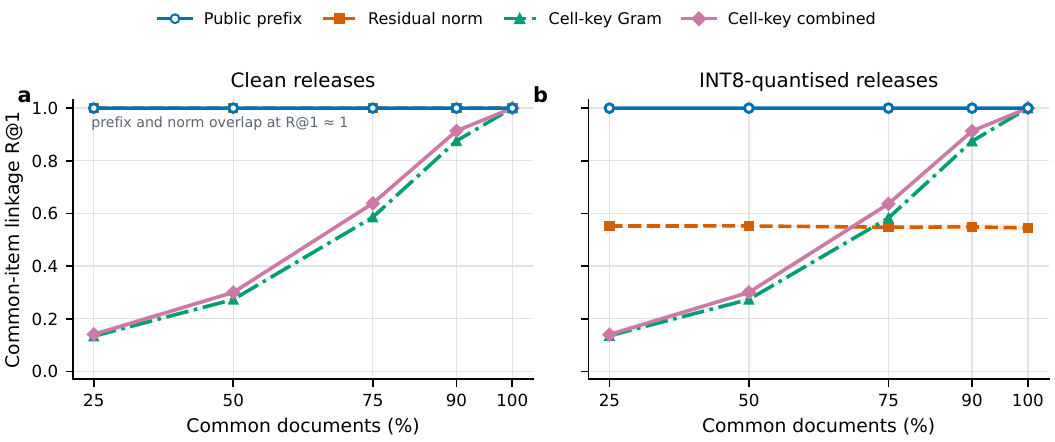}
\caption{Linkage of persistent e5-base records across independently re-keyed
$N=10{,}000$ releases. Cell-key Gram structure weakens at low overlap and
recovers as the common population grows; the stable prefix remains almost
perfect. Int8 reduces exact norm R@1 but leaves a strong pairwise signal.
Bands show one standard deviation over release pairs.}
\label{fig:shard-churn}
\end{figure}

An open-set, threshold-free mutual-nearest-neighbour rule applied to the
public-prefix distance makes the churn effect visible in a different way. Its
common-row recall stays essentially one, while precision is only
$0.583$/$0.616$ at $25\%$ overlap and rises to about $0.996$ at $90\%$. The
low-overlap errors are false reciprocal matches
among churn-only rows, not protection of persistent documents. Per-document
micro-keys do destroy cross-release Gram invariance (roughly $0.016$--$0.021$
R@1), yet they still leave the prefix and residual norm. The expanded game
therefore narrows the earlier statement---Gram linkage depends on population
overlap---while strengthening the central conclusion that SHARD releases are
not unlinkable. Because clean and fp16 persistent prefixes also permit a
near-zero distance threshold, this threshold-free MNN diagnostic is
conservative for the unchanged-prefix channel.

\subsection{A formally calibrated Gaussian release}\label{ssec:shard-vs-dp}

The earlier isotropic-noise sweep remains only a distortion diagnostic: it
had no adjacency relation, clipping rule or privacy accountant. Experiment~27
replaces that misuse of terminology with an explicit one-shot mechanism.
Databases have fixed size and are adjacent when one row is replaced; row
participation and identifiers are public. Every embedding is clipped to
$B=1.000001$, a fixed bound from the encoder's unit-normalisation contract
rather than a corpus estimate. The concatenated release therefore has global
$L_2$ sensitivity $\Delta_2=2B=2.000002$. Independent Gaussian noise is added
to every coordinate and each row is then unit-normalised as post-processing.
For $\delta=10^{-6}$, $\sigma$ is obtained from the exact analytic Gaussian
privacy-loss equation of Balle and Wang~\cite{ref_analytic_gaussian}; all five
executable calibration checks recover the requested $\delta$ to numerical
precision.
Writing $\Phi$ for the standard-normal CDF, the solved equation is
\[
\delta=\Phi\!\left(\frac{\Delta_2}{2\sigma}
                 -\frac{\varepsilon\sigma}{\Delta_2}\right)
-e^{\varepsilon}\Phi\!\left(-\frac{\Delta_2}{2\sigma}
                 -\frac{\varepsilon\sigma}{\Delta_2}\right).
\]

This gives central $(\varepsilon,\delta)$-DP for document-vector content under
the stated replacement adjacency when a trusted curator publishes the joint
release. The same per-record calibration has a local interpretation only if
each record owner clips and perturbs its own vector. It is not add/remove
membership privacy, does not protect queries, and repeated independent
releases compose. Centroids or PCA fitted on private release data would need
their own mechanism and privacy accounting; here preprocessing parameters are
public or independently trained.

We evaluate eight real BEIR/MIRACL encoder--dataset cases, 17 finite
$\varepsilon$ values, three noise seeds, every judged query, and 256 linkage
targets per case against a clean native gallery of up to $5{,}000$ rows after
identifiers are removed. Confidence intervals use 2,000 hierarchical
bootstrap repetitions. Table~\ref{tab:formal-dp} summarises three operating
points. At $\varepsilon=1$, nDCG@10 is at most $0.011$. By
$\varepsilon=512$, native-gallery linkage is already $0.962$--$0.995$ while
nDCG@10 is only $0.017$--$0.062$. A strict utility match to corrected SHARD,
defined by $|\Delta\mathrm{nDCG@10}|\le0.01$ and
$|\Delta\mathrm{Recall@100}|\le0.02$, occurs in only three of eight cases on
the finite grid, all at $\varepsilon=32768$ and all with linkage R@1 at least
$0.995$.

\begin{table}[tbp]
\centering
\caption{Analytic Gaussian release at $\delta=10^{-6}$. Metric ranges are
across eight evaluated cases and are rounded to three decimals; ``matches''
counts simultaneous matches to the corrected SHARD nDCG@10 and Recall@100
targets on the measured grid.}
\label{tab:formal-dp}
\small
\begin{tabular}{@{}rrrrrr@{}}
\toprule
$\varepsilon$ & $\sigma$ & nDCG@10 & Recall@100 & link R@1 & matches \\ \midrule
$1$ & $8.4494$ & $.000$--$.010$ & $.000$--$.028$ & $.000$ & $0/8$ \\
$512$ & $.07240$ & $.017$--$.062$ & $.071$--$.244$ & $.962$--$.995$ & $0/8$ \\
$32768$ & $.007959$ & $.297$--$.675$ & $.250$--$.964$ & $.995$--$1.000$ & $3/8$ \\
\bottomrule
\end{tabular}
\end{table}

\begin{figure}[tbp]
\centering
\includegraphics[width=0.98\linewidth]{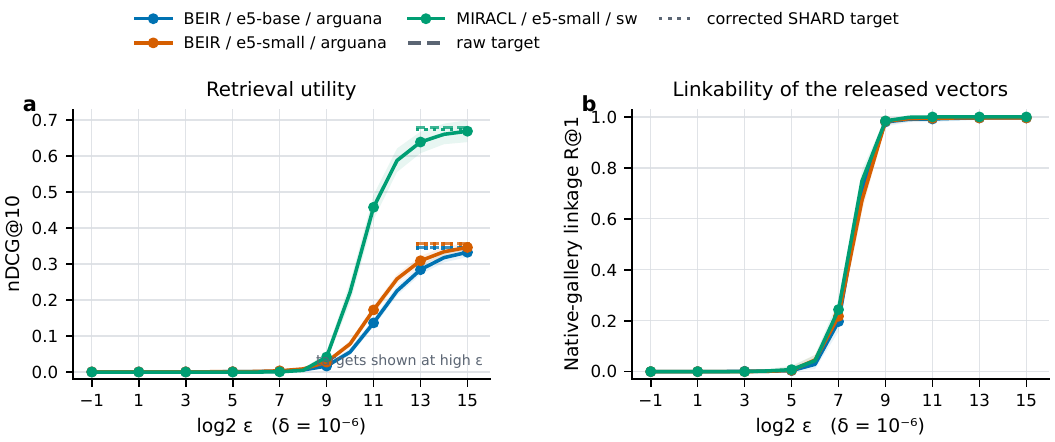}
\caption{Utility and native-gallery linkability of the formal Gaussian
release. Representative datasets are selected from non-private utility before
examining any finite-$\varepsilon$ result. Strong noise destroys retrieval;
as utility returns, clean-reference linkage rises sharply.}
\label{fig:formal-dp}
\end{figure}

The comparison is deliberately limited to retrieval quality. The Gaussian
baseline scans the exact full corpus, whereas corrected SHARD uses two-stage
$K=200$ routing; latency and traffic are therefore not matched. More
importantly, SHARD is not DP, and an empirical linkage rate is not a privacy
budget. The useful conclusion is narrower: for this simple per-record
Gaussian release, every strict utility match occurs at
$\varepsilon=32768$ and has native-gallery linkage R@1 at least $0.995$. It is
a formal baseline with a different guarantee, not a proof that geometric
keying dominates differential privacy.

\paragraph{A measured limitation: overlap reference lookup.}\label{ssec:shard-ref}

The baseline allowed an overlapping reference corpus to recover exact
source paragraphs after alignment ($99.8\%$ top-1,
Section~\ref{ssec:reference-attack}). The actual \textsc{Shard} construction
does not key the prefix: $u_i$ is the stable stage-1 representation. Thus an
attacker with an overlapping native reference does not need to recover a
prefix key at all. The earlier optional keyed-prefix ablation reached
$0.995$ R@1 after $200$ anchors on a $100\,000$-document e5-small reference;
the unkeyed construction exposes at least as direct a channel, consistent
with the $0.999$ cross-release prefix linkage in Experiment~25. Residual
cell keys and micro-keys do not close it. A coarse-cell-ID public channel or
a client-private routing index could reduce this exposure, but their recall,
storage and onboarding costs are not evaluated here.

\FloatBarrier
\section{Discussion}\label{sec:discussion}

\paragraph{What compartmentalisation buys---and what it does not.} Cell keys
remove one globally shared residual orientation. Under diffuse disclosure,
that shifts the required \emph{global} pair budget because each cell sees
only a fraction of the pairs. Experiment~24 shows the shift directly, but it
also shows why it is not a privacy threshold: minimum-norm OLS uses the
observed anchor span and re-identifies records from roughly $32$--$36$
in-cell pairs, far below $d_{\mathrm{priv}}=288$. Micro-keys remove shared
residual edges, but norms and the prefix remain stable. Thus $C$ is a
compartmentalisation--cost knob, not a security parameter. The positive
result that survives is useful but narrow: corrected scoring preserves raw
ranking, and diffuse evidence accumulates more slowly per cell. The negative
result is equally central: the stored view remains linkable and is not a
cancellable template.

\paragraph{Encoder selection is part of the protocol.} The baseline
protective wrapper is metric-preserving only for the score it actually
implements. Once query centring is corrected, the half-PCA gap is strongly
encoder- and dataset-dependent: it is negligible for e5-base/NFCorpus but
reaches $0.053$ nDCG on MIRACL/Swahili with e5-small. Encoder and projection
width should therefore be selected from corrected, graded relevance metrics
on the deployment distribution; neither dimension nor $\sigma_{rec}$ is a
privacy or utility guarantee.

\paragraph{The role of the rotation $R$: obfuscation, not a primitive.}
We state this as plainly as possible: \emph{the secret rotation is an
empirical obfuscation layer, not a cryptographic primitive}. It
contributes a defence channel orthogonal to SVD truncation, and even at
$k/d = 1.0$ the unknown-$R$ regime reduces off-the-shelf Vec2Text BLEU
to the noise floor---but the paper also reports that a known-$R$ attacker
gains nothing beyond SVD and that a known-plaintext attacker recovers the
orientation with standard Procrustes alignment. The aligned Vec2Text
stress test does not recover text or typed PII, but its raw baseline is
already near the reconstruction floor. Experiment~29 repairs that control in
a separate, checkpoint-compatible GTR case: an unknown cell key lowers
off-the-shelf reconstruction, but the public prefix and a small known-pair span
restore part of it. Because that test grants $V$ and $\mu$, uses native
unnormalised GTR vectors and does not train on SHARD outputs, a corpus-adapted
or universal decoder is still untested. The rotation's value proposition is
therefore limited: it is a cheap obstacle
for a weak unknown-orientation attacker, not a security guarantee. The
scientific weight of the paper is the measured privacy/latency/accuracy
trade-off and the documented failure modes, not the rotation.

\paragraph{What CKKS hides and what it does not.} In the baseline prototype,
CKKS hides the numerical scoring query and approximate score values from an
honest-but-curious evaluator that lacks the secret key. Experiment~26 now
executes that path for SHARD itself and shows that CKKS approximation preserves
the measured order, but implementation evidence does not enlarge the security
claim. CKKS does \emph{not} hide active cells, ciphertext counts, the
identifiers of candidates produced by stage-1 filtering, or the distances
implicit in their order. Section~\ref{ssec:alignment-pq} also shows that public PQ
codes preserve substantial local geometry even before observing queries.
An access-pattern attacker can therefore, over many queries, build a
co-occurrence model of which documents are co-retrieved.
Composing the pipeline with a PIR-style retrieval primitive (Tiptoe
\cite{ref_tiptoe}, SealPIR \cite{ref_sealpir}, OnionPIR
\cite{ref_onionpir}, SimplePIR \cite{ref_simplepir}) is the natural
remedy; ours is engineered to be PIR-friendly because the
ct-pt operations on the server are stateless and cluster-local.

\paragraph{Parameter selection.} In the earlier baseline, the selected configuration is
$\approx 1.7\times$ faster than the TenSEAL stock parameters at the same
parameter-table security bound and within $\Delta\mathrm{Acc@1}\leq 1$~p.p.\
of baseline\_proj. We reiterate that this speed-up is structural---ct-pt
needs one multiplicative level, hence a shorter modulus chain---rather than
a product of the learned surrogate, which only serves to extrapolate
latency to unbenchmarked configurations or hardware. The selection runs
offline at no per-query cost. These timings cannot be transferred directly
to SHARD because its residual dimension, active-cell multiplicity and packing
layout differ.

\section{Limitations and Future Work}\label{sec:limits}

\paragraph{Intrinsic limitations.}
The most important limitation is now measured rather than hypothetical:
cell-local orthogonal keys compartmentalise known pairs but do not create a
de-anonymisation threshold. An observer needs full anchor rank to identify a
generic complete key, yet a minimum-norm map reaches high residual-gallery
R@1 from roughly $32$--$36$ in-cell pairs. A targeted observer can collect
those pairs in one public cell, so increasing the global cell count does not
improve this per-target case.

The stored representation is also linkable across key epochs. The public
prefix is unchanged, every orthogonal transform preserves residual norm, and
cell keys preserve the complete within-cell Gram matrix. Micro-keys remove the
last shared-geometry channel but retain prefix and norm. Experiment~28 extends
the same-population audit to two encoders, partial overlap and insert/delete
churn. It still keeps the embedding of each persistent document fixed; text
edits, stochastic re-embedding and encoder-version drift are not simulated.
The numerical linkage rates should not be extrapolated to every deployment,
but they are sufficient to reject unlinkability and cancellable-template
claims for the stated construction.

The systems evidence is now concrete but still local. Experiment~26 measures
SHARD transforms, CKKS encryption/decryption, ct--pt evaluation,
serialization, point-sampled RSS, throughput, score error and ranking flips.
It excludes preprocessing from online latency and does not include a real
network, TLS, ANN service, concurrent clients or a packed multi-score response.
TenSEAL executes this path on the CPU; the RTX~5060 does not accelerate CKKS.
Most importantly, CKKS does not hide candidate identifiers, active cells,
ciphertext counts or repeated-query access patterns.

The empirical scope remains bounded. The corrected utility audit covers two
encoders and ten BEIR/MIRACL cells, with the largest graded corpus at about
$297\,000$ passages. Partial alignment still uses e5-small residuals, while
the churn audit covers e5-small and e5-base. The learned inversion experiment
uses one GTR/Vec2Text checkpoint family, two fixed geometry designs and twelve
targets per cell. To stay on the corrector's input distribution it uses native
unnormalised GTR vectors and deliberately grants the observer the exact PCA
basis and mean; its bootstrap intervals are conditional descriptions, not
population-wide significance claims. It cannot represent every generative or
universal inverter. Membership inference and non-overlapping semantic
reference lookup are not evaluated. The formal Gaussian release covers one
simple mechanism under public participation and replacement adjacency; it is
not an add/remove membership guarantee and does not privately learn PCA.

\paragraph{Priorities for future work.}
The first systems priority has shifted from implementing block-SIMD to fixing
the response path it exposed. A useful client/server prototype should pack
multiple candidate scores per response, include ANN routing and real
RTT/TLS, and report throughput and tail latency under concurrent load. The
present phase breakdown supplies a reproducible single-machine baseline for
that work rather than a production-service forecast.

The first construction priority is to remove the invariants exposed by
Experiments~25 and~28. One candidate is a well-conditioned asymmetric cell transform
$z_i=A_c r_i$, $t_{q,c}=A_c^{-\top}r_q$, which preserves
$t_{q,c}^\top z_i=r_q^\top r_i$ without necessarily preserving
document--document norms or Gram geometry. This direction is related to
asymmetric scalar-product-preserving encryption and must be positioned as a
systems composition rather than a new cryptographic primitive. Any variant
must be retested under the full server view, multiple release epochs,
text/embedding drift, insert/delete churn and known-pair alignment.

Finally, the learned text experiment should be widened to corpus-adapted and
universal inverters, more encoder families, entity/PII categories, membership
inference and semantic lookup against non-overlapping corpora. On the formal
privacy side, stronger baselines should include private preprocessing,
add/remove membership adjacency and explicit multi-release composition rather
than relabelling empirical non-linkage as DP. Both lines need a million-scale
graded retrieval corpus before deployment-oriented claims are credible.

\section{Conclusion}\label{sec:conclusion}

This study revises both the utility and the security interpretation of
cell-keyed residual splitting. On utility, the decisive correction is simple:
fit PCA on centred documents, retain centred coordinates for routing, but use
the uncentred query for scoring. The resulting full score differs from the raw
dot product only by a query-dependent constant. Ten BEIR/MIRACL cells confirm
the identity and show that the earlier query-centring mismatch, not truncation
alone, caused a substantial part of the reported baseline loss.

On security, cell keys do compartmentalise diffuse known-pair evidence, but
the effect is gradual. Minimum-norm alignment recovers the anchor-span
projection and exceeds $0.9$ residual-gallery R@1 with about $32$--$36$
in-cell pairs while the full orthogonal key is still underdetermined. The
global disclosure budget grows as those pairs are spread over more cells,
yet there is no hard $d_{\mathrm{priv}}$ de-anonymisation threshold and no
corresponding secure choice of $C$.

The strengthened GTR outcome test gives the same conclusion in generated
text. With the exact PCA frame disclosed but no cell-key anchors, mean
off-the-shelf token-F1 falls from $.665$ to $.242$. The public prefix alone
retains $.433$, however, and eight-pair OLS reaches $.450$. This is a bounded
no-anchor benefit, not resistance to targeted or trained reconstruction.

Re-keying also fails the stronger unlinkability game. Residual norms,
within-cell Gram signatures and the unchanged public prefix link almost every
record across independently permuted releases. Insert/delete churn weakens
Gram matching when overlap is low, but not the stable prefix or norm;
micro-keys remove shared residual edges without closing those channels. The
formal Gaussian baseline also shows why a different guarantee cannot be
summarised by one empirical noise knob: useful retrieval returns only where
clean-reference linkability is already high in the evaluated grid.

The implementation result is similarly specific. Real block-SIMD CKKS keeps
the score order and sharply reduces query upload, but the current TenSEAL
layout is slower because it emits one large encrypted response per candidate.
SHARD should therefore be understood as a rank-preserving
alignment-compartmentalisation mechanism with explicit leakage and a measured
cryptographic query path, not as an unlinkable, cancellable, DP or
cryptographically private document store. Breaking the observed invariants
and packing the response path are the next construction and systems
requirements before deployment-oriented privacy claims are warranted.

\section*{Data and code availability}

The standalone SHARD repository at
\url{https://github.com/sergkurilenko/shard-private-dense-retrieval}
contains the implementation, experiment configurations, curated
machine-readable outputs, run metadata, figure-generation code, and the
canonical and JISA manuscript sources. The SVD/rotation/PQ comparison
baseline is maintained separately at
\url{https://github.com/sergkurilenko/hybrid-privacy-aware-semantic-search};
this repository retains only the three static baseline figures needed by the
comparison in this article.

Large embedding caches and third-party corpora are not redistributed. The data
manifest records the source datasets, model revisions, deterministic sampling
rules, cache formats, and commands needed to rebuild them. For Experiment~29,
the public artifact includes configurations, sample indices, content hashes,
numeric case aggregates, bootstrap summaries, and runtime metadata, while raw
AG News text and decoded case-level text are deliberately omitted. Synthetic
PII controls and aggregate outcomes remain available for auditing the stated
scope.

\appendix
\section{Reproduced baseline analyses from the global-linear system}\label{app:baseline}\label{ssec:vec2text}\label{app:projbaselines}\label{app:noise}\label{ssec:adaptive-inversion}

The security analyses of the global-linear baseline are the primary
contribution of our prior work and are reported there in
full~\cite{ref_alpha}; we do not reproduce them here and summarise
only the conclusions the main text relies on. \emph{Off-the-shelf
\textsc{Vec2Text}} against SVD-truncated and rotated embeddings: an unknown
rotation drops inversion BLEU to the noise floor, but \emph{only} against a
non-adaptive inverter---a known rotation, or the known-plaintext alignment
of Section~\ref{ssec:alignment-pq}, removes the protection.
\emph{Lightweight projection baselines} (Gaussian and random-orthogonal):
the choice of projection is an encoder-dependent utility trade-off, not a
universal privacy primitive. \emph{Matched-distortion Gaussian diagnostic
(not a DP mechanism)}:
matched-distortion noise can preserve more raw-neighbour structure than SVD
at the same $\sigma_{rec}$, so $\sigma_{rec}$ is not a privacy metric.
\emph{Aligned \textsc{Vec2Text} stress test} (Procrustes alignment then a
generative corrector): no exact text or typed PII is recovered---but neither
is it from the raw control, so this test bounds rather than certifies
resistance. These are exactly the global-key weaknesses that motivate
\textsc{Shard}.

\end{document}